\documentclass{llncs}

\usepackage{multicol}
\usepackage{amsmath}
\usepackage{amsfonts}
\usepackage{amssymb}
\usepackage[vario]{fancyref}
\usepackage{verbatim}
\usepackage{fancybox}
\usepackage{algorithm}
\usepackage{qtree,algorithmic}
\usepackage{answers}
\usepackage{xspace}
\usepackage{comment}
\usepackage{color}
\usepackage[T1]{fontenc}

\newcommand{\bitem}{\begin{itemize}}
\newcommand{\eitem}{\end{itemize}}

\newcommand{\N}{\mathbb N}

\newcommand{\Z}{\mathbb Z}
\newcommand{\Ninf}{\N \cup \{ \infty \}}

\newcommand{\isdef}{\stackrel{\mbox{\tiny def}}{=}}

\newcommand{\trueform}{\top}
\newcommand{\falseform}{\bot}
\newcommand{\iseq}{\simeq}
\newcommand{\lesseq}{\preceq}

\newcommand{\less}{\prec}

\newcommand{\truevalue}{{\mathrm{true}}}
\newcommand{\falsevalue}{{\mathrm{false}}}

\newcommand{\B}{{\cal B}}
\newcommand{\C}{{\cal C}}
\newcommand{\D}{{\cal D}}

\newcommand{\I}{{\cal I}}

\newcommand{\V}{{\cal V}}
\newcommand{\T}{{\cal T}}

\renewcommand{\P}{{\cal P}}
\newcommand{\Q}{{\cal Q}}

\renewcommand{\S}{{\cal S}}

\newcommand{\profile}{\mathrm{\it profile}}

\newcommand{\sort}{\tt s}
\newcommand{\su}{\mathrm{succ}}
\newcommand{\terms}[1]{T_{#1}}

\newcommand{\nat}{{\tt nat}}
\newcommand{\lists}{{\tt list}}
\newcommand{\integers}{{\tt int}}

\newcommand{\bool}{\tt bool}
\newcommand{\intof}[2]{{#1}^{#2}} 

\newcommand{\valof}[2]{[#1]^{#2}} 

\newcommand{\replaceby}[3]{#1[#3/#2]}

\newcommand{\nil}{{\text \it nil}}
\newcommand{\cons}{{\text \it cons}}

\newcommand{\NNF}{\text{\it NNF}}

\newcommand{\ordrel}{\triangleleft}

\newcommand{\antispace}{\hspace*{-0.1cm}}

\newcommand{\even}{\text{\it even}}

    \newcommand{\addproof}[2]{\Writetofile{appfile}{\protect\section{Proof of ##2} ##1}}
    \newcommand{\addtheowithlemmas}[4]{##1 \Writetofile{appfile}{\protect\section{Proof of ##2} ##3 \protect\section*{Main proof} ##4}}
    \newcommand{\docbegin}{\Opensolutionfile{appfile}}
    \newcommand{\docend}{\Closesolutionfile{appfile} \include{appfile}}
    \newcommand{\aboutproofs}[1]{All proofs can be found in the Appendix.}
    \newcommand{\proofstate}[2]{##2}
    \newcommand{\specthanks}{\thanks{This technical report is the preliminary version of
a paper accepted for presentation in CICM 2012 (Conferences on Intelligent Computer Mathematics) and included
in the proceedings of the conference published by Springer in their Lecture Notes in Artificial Intelligence series.
The final publication is available at
www.springerlink.com. This work has been partly funded by the project ASAP of
the French {\em Agence Nationale de la Recherche} (ANR-09-BLAN-0407-01).}}

\title{Reasoning on Schemata of Formul{\ae}\specthanks}
\author{Mnacho Echenim and Nicolas Peltier}
\institute{University of Grenoble\thanks{emails:
      \texttt{Mnacho.Echenim@imag.fr},
      \texttt{Nicolas.Peltier@imag.fr}
    }
    (LIG, Grenoble INP/CNRS)}

\date{}

\begin{document}

\maketitle

\begin{abstract}
A logic is presented for reasoning on iterated sequences of formul{\ae} over some given base language. The considered sequences, or \emph{schemata}, are defined inductively, on some algebraic structure (for instance the natural numbers, the lists, the trees etc.).
A proof procedure is proposed to relate the satisfiability problem for schemata
to that of finite disjunctions of base formul{\ae}.
It is shown that this procedure is sound, complete and terminating, hence the basic computational pro\-per\-ties of the base language can be carried over to schemata.
\end{abstract}

\docbegin

\newcommand{\free}{non-instantiated\xspace}
\newcommand{\instantiated}{instantiated\xspace}
\newcommand{\nonisolated}{$\D$-dependant\xspace}
\newcommand{\isolated}{$\D$-independant\xspace}

\newcommand{\parentd}[1]{\triangleright_{#1}}
\newcommand{\parentn}[2]{\geq_{#1,#2}}
\newcommand{\eq}[1]{\mathrm{Eq}(#1)}
\newcommand{\purify}[1]{\mathrm{RmEq}(#1)}
\newcommand{\noneq}[1]{\mathrm{NonEq}(#1)}
\newcommand{\purifybis}[2]{\mathrm{RmEq}'(#1,#2)}

\section{Introduction}

We introduce a logic for reasoning on iterated schemata of formul{\ae}. The schemata we consider are infinite sequences of formul{\ae} over a given {\em base language}, and these sequences are defined by induction on some algebraic structure (e.g. the natural numbers).
As an example, consider the following sequence of propositional formul{\ae} $\phi_n$, parameterized by a natural number $n$:
\[\mbox{$\phi_0 \rightarrow \trueform$ \qquad
$\phi_{n+1} \rightarrow \phi_n \wedge (p(n) \Leftrightarrow p(n+1))$.}\]
It is clear that the formula $\phi_n \wedge p(0) \wedge \neg p(n)$ is unsatisfiable, for every $n \in {\Bbb N}$.
This can be easily checked by any SAT-solver, for every {\em fixed} value of $n$.
Here the base language is propositional logic and the sequence is defined over the natural numbers.
 However, proving that it is is unsatisfiable \emph{for every $n \in {\Bbb N}$} is a much harder task which obviously requires the use of mathematical induction.
Similarly, consider the sequence:
\[\mbox{$\psi_\nil \rightarrow \trueform$\qquad
$\psi_{\cons(x,y)} \rightarrow \psi_y \wedge (\exists u\, p(y,u)) \Leftrightarrow (\exists v\, p(\cons(x,y),v))$}\]
Then $\psi_l \wedge p(\nil,a) \wedge \forall u\, \neg p(l,u)$ is unsatisfiable, for every (finite) list $l$.
Here the base language is first-order logic and the sequence is defined over the set of lists.
Such inductively defined sequences are ubiquitous in mathematics and computer science. They are often introduced to analyze the complexity of proof procedures. From a more practical point of view, schemata of propositional formul{\ae} are used to model properties of circuits parameterized by natural numbers, which can represent, e.g., the number of bits, number of layers etc. (see for instance \cite{DBLP:conf/cav/GuptaF93}, where a language is introduced to denote inductively defined boolean functions which can be used to model such parameterized circuits). In mathematics, schemata of first-order formul{\ae} can model
inductive proofs, which can be seen as infinite (unbounded) sequences of first-order formul{\ae} (see \cite{DBLP:journals/tcs/BaazHLRS08} for an example of the use of this technique in proof analysis).

\newcommand{\maxs}{\text{\it max}}

\newcommand{\mulout}{\text{\it out}}
\newcommand{\mulin}{\text{\it signal}}
\newcommand{\mulch}{\text{\it sel}}

\newcommand{\mulind}{\text{\it Ind}}
\newcommand{\mulbase}{\text{\it Base}}

We now provide a slightly more complex example. The following schema $\psi_{t}$ encodes a multiplexer, inductively defined as follows.
The base case is denoted by $\mulbase(x)$, where $x$ denotes an arbitrary signal. In this case, the output of the circuit is simply the output of $x$, denoted by $\mulin(x)$.
The inductive case is denoted by $\mulind(i,x,y)$, where $i$ is a select input and $x$ and $y$ are two smaller instances of the multiplexer. Its output is either the output of $x$ or that of $y$, depending on the value of $i$.
\[\begin{tabular}{llll}
$\psi_{\mulbase(x)}$ & $\rightarrow$ & \multicolumn{2}{l}{$\mulout(\mulbase(x)) \Leftrightarrow \mulin(x)$} \\
$\psi_{\mulind(i,x,y)}$ & $\rightarrow$ &
 & $\left(\neg \mulin(i) \vee \left(\mulout(\mulind(x,y)) \Leftrightarrow \mulout(x)\right)\right)$ \\
& & $\wedge$ & $\left(\mulin(i) \vee \left(\mulout(\mulind(x,y)) \Leftrightarrow \mulout(y)\right)\right)$ \\
& & $\wedge$ & $\psi_{x} \wedge \psi_{y}$ \\
\end{tabular}\]

Note that this kind of circuit cannot be encoded in the language of (regular) propositional schemata defined in \cite{ACP09a,ACP10}, because the number of inputs
is exponential in the depth of the circuit. Hence, the use of non-monadic function symbols is mandatory.

In this paper, we devise a proof procedure to check the satisfiability of these sequences.
More precisely, we introduce a formal language for modeling sequences of formul{\ae} defined over an arbitrary base language (encoded as first-order formul{\ae} interpreted in some particular theory) and we show that the computational properties of the base logic carry over to these schemata: If the satisfiability problem is decidable (resp. semi-decidable) for the base language then it is also decidable (resp. semi-decidable) for the corresponding schemata.
For instance, the satisfiability problem is decidable for schemata of propositional formul{\ae} and semi-decidable for schemata of first-order formul{\ae}.
The basic principle of our proof procedure consists in relating the satisfiability of any iterated schemata of formul{\ae} to that of a \emph{finite} disjunction of base formul{\ae}.
 The complexity of the satisfiability problem, however, is not preserved in general, since the number of formul{\ae} in the disjunction may be exponential.

This work generalizes previous results \cite{ACP09a,ACP10} in two directions: first the base language is no longer restricted to propositional logic\footnote{A first extension to some decidable theories such as Presburger arithmetic was considered in \cite{AP11a}.} and second the sequences are defined over arbitrary algebraic structures, and not only over the natural numbers.
 Abstracting from the base language leads to an obvious gain in applicability since our approach now applies to any logic, provided a
 proof procedure exists for testing the satisfiability of base formul{\ae}. Besides, it has the advantage that the reasoning on schemata is now clearly separated from the reasoning on formul{\ae} in the base language, which may be postponed.
This should make our approach much more scalable, since any existing system could now be used as a ``black box'' to handle the basic part of the reasoning (whereas the two aspects were closely interleaved in our first approach, yielding additional computational costs).
Both extensions significantly increase the scope of our approach.

The extension to arbitrary structures turns out to be the most difficult from a theoretical point of view, mainly because, as we shall see, the number of parameters can increase during the decomposition phase, yielding an increase of the number of related non-decomposable formul{\ae} in each branch, which can in principle prevent termination. In contrast to what happens in the simpler case of propositional schemata \cite{ACP09a}, these formul{\ae} \emph{cannot} in general be deleted by the purity principle, since they are not independent from the other formul{\ae} in the branch.
To overcome this problem, we devise a specific instantiation strategy based on a careful analysis of the depth of terms represented by the parameters, and we define a new loop detection mechanism.
This blocking rule is more general and more complex than the one in \cite{ACP09a}. We show that it is general enough -- together with the proposed instantiation strategy -- to ensure  termination. Termination is however much more difficult to prove than for propositional schemata defined over natural numbers.

The types of structures that can be handled are quite general: they are defined by sets of --  possibly non-free -- constructors on a sorted signature. The terms can possibly contain elements of a non-inductive sort. For instance, a list may defined inductively on an \emph{arbitrary} set of elements.

\subsection*{Related Work}

There exist many logics and frameworks in which the previous schemata can be encoded, for instance higher-order logic \cite{Benzmuller:2008:LCA:1431108.1431127}), first-order $\mu$-calculus \cite{firstorder_mucalculus}, or logics with inductive definitions \cite{handbook_aczel} that are widely used in proof assistants \cite{coq_definitions}.
However, the satisfiability problem is not even semi-decidable for these logics (due to G\"odel's famous result).
Very little published research seems to be
focused on the identification of complete subclasses
and iterated schemata definitely do not lie in these classes and
cannot be reduced to them either.
Our approach ensures that the basic computational properties of the base language (decidability or semi-decidability) are preserved, at the cost of additional restrictions on the syntax of the schemata under consideration.
Furthermore, the modeling of schemata in higher-order languages, although possible from a theoretical point of view, is cumbersome and not very natural in practice.

There exist several approaches in inductive theorem proving, ranging from explicit induction approaches (see for instance \cite{DBLP:books/el/RV01/Bundy01} or \cite{DBLP:conf/cade/BaeldeMS10})  used mainly by proof assistants
to implicit induction schemes used in rewrite-based theorem provers \cite{663845,Bouhoula95implicitinduction}, or even to inductionless induction \cite{KM87,COM01}, where inductive validity is reduced to a mere satisfiability check.
Such approaches can in principle handle some of the formul{\ae} we consider in the present work, provided the base language can be axiomatized. Existing approaches are usually only complete for refutation, in the sense that false conjectures can be disproved, but that
inductive theorems cannot always be recognized (this is theoretically unavoidable). Once again, very few termination results exist for such provers and our language does not fall in the scope of the known complete classes (see for instance \cite{DBLP:conf/cade/GieslK01}).
In general, inductive theorem proving requires strong human guidance, especially for specifying the needed inductive lemmata. In contrast, our procedure is \emph{purely automatic}. Of course, this comes at the expense of strongly reducing the form of the inductive axioms.
Furthermore, although very restricted to ensure termination and/or completeness, our language allows for more general queries, possibly containing nested quantifiers, which are in general out of the scope of existing automated inductive theorem provers.
Indeed, most existing approaches aim at establishing the inductive validity of universal queries w.r.t. a first-order axiomatization (usually a set of clauses). In contrast, our method can handle more general goals of the form $\forall \vec{x}\, \phi$, where $\vec{x}$ is a vector of variables interpreted over the considered algebraic structure and $\phi$ is a formula containing arbitrary quantifiers in the base language.

Practical attempts to use existing inductive theorem provers (such as ACL \cite{DBLP:conf/cade/BoyerM90}) to check the satisfiability of schemata such as those in the Introduction fail for every formula except the most trivial ones. We believe that this is not due to a lack of efficiency, but rather to the fact that additional inductive lemmata are required, which cannot be generated automatically by the systems. In some sense, our method (and especially the loop detection rule) can be viewed as an automatic way to generate such lemmata.
Our method is also more modular: we make a clear distinction between the reasoning over the base logic and the one over inductive definitions.
 Inference rules are devised for the latter and an external prover is used to establish the validity of formul{\ae} in the base language.

Since parameterized schemata can obviously be seen as monadic predicates, a seemingly natural idea would be to encode them in monadic second-order logic and use an automata-based approach (see, e.g., \cite{springerlink:10.1007/3-540-44693-1_41}) to solve the satisfiability problem.
However, as we shall see in Section \ref{sect:proof}, the unfolding of the inductive definitions contained in a given formula may well increase the number of parameters occurring in it. Since these parameters may share subterms, the formul{\ae} containing them are \emph{not} independent hence they must be handled simultaneously, in the same branch. Thus a systematic decomposition into monadic atoms (in the style of automata-based approaches) is not feasible.


 \aboutproofs{}

\section{A Logic for Iterated Schemata}

\label{sect:prel}

\newcommand{\indsort}{\I}
\newcommand{\propof}[2]{#2|_{#1}}
\newcommand{\propin}[1]{\Q(#1)}
\newcommand{\bad}{$\moregen$-bad\xspace}
\newcommand{\good}{$\moregen$-good\xspace}
\newcommand{\similar}[1]{\sim_{#1}}
\newcommand{\paramof}[2]{\param(#1,#2)}

\newcommand{\inductivesort}{inductive sort\xspace}
\newcommand{\inductivesorts}{inductive sorts\xspace}
\newcommand{\aninductivesort}{an inductive sort\xspace}
\newcommand{\anoninductivesort}{a non inductive sort\xspace}

\newcommand{\const}{\C}
\newcommand{\param}{\P}
\newcommand{\indsymb}{$\indsort$-symbol\xspace}
\newcommand{\indsymbs}{$\indsort$-symbols\xspace}

\newcommand{\iparameters}{parameters\xspace}
\newcommand{\iparameter}{parameter\xspace}
\newcommand{\ind}[1]{\indterm(#1)}
\newcommand{\indterm}{\terms{\indsort}}
\newcommand{\iterms}{$\indsort$-terms\xspace}
\newcommand{\iterm}{$\indsort$-term\xspace}

The schemata we consider in this paper are encoded as first-order formul{\ae}, together with a set of rewrite rules specifying the interpretation of certain monadic predicate symbols.
Our language is \emph{not} a subclass of first-order logic: indeed, some sort symbols will be interpreted on an inductively defined domain (e.g. on the natural numbers).
Furthermore, the formul{\ae} can be interpreted modulo some particular theory, specified by a class of interpretations.

We first briefly review usual notions and notations.
We consider first-order  terms and formul{\ae} defined on a sorted signature.
Let $\S$ be a set of \emph{sort symbols}.
Let $\Sigma$ be a set of \emph{function symbols}, together with a function $\profile$ mapping every symbol in $\Sigma$ to a unique non-empty sequence of elements of $\S$. We write $f: \sort_1 \times \dots \times \sort_n \rightarrow \sort$ if $\profile(f) = \sort_1,\dots,\sort_n,\sort$ with $n > 0$, and $a: \sort$ if $\profile(a) = \sort$ (in this case $a$ is a \emph{constant symbol}).
A symbol is \emph{of sort $\sort$} and \emph{of arity $n$} if its profile is of the form $\sort_1,\dots,\sort_n,\sort$ (possibly with  $n = 0$). The set of function symbols of sort $\sort$ is denoted by $\Sigma_{\sort}$.
Let $(\V_{\sort})_{\sort \in \S}$ be a family of pairwise disjoint set of \emph{variables of sort $\sort$}, and $\V \isdef \bigcup_{\sort \in \S} \V_{\sort}$.
We denote by $\terms{\sort}$ the sets of \emph{terms of sort $\sort$} built as usual on $\Sigma$ and $\V$.
A term not containing any variable is \emph{ground}.

\begin{definition}
Let $\indsort$ be a subset of $\S$. The elements of $\indsort$ are called the {\em \inductivesorts.}
An {\em \iterm} is a term of a sort $\sort \in \indsort$.

Let $\const \subseteq \Sigma$ be a set of \emph{constructors}, such that
the sort of every symbol in $\const$ is in $\bigcup_{\sort \in \indsort} \Sigma_{\sort}$ and such that
every non-constant symbol of a sort in $\bigcup_{\sort \in \indsort} \Sigma_{\sort}$ is in $\const$.
A \emph{parameter} is a constant symbol of a sort occurring in the profile of a constructor (parameters are denoted by upper-case letters).
A term containing only function symbols in $\const$ and variables of sorts in $\S \setminus \indsort$ is a \emph{constructor term}.
\end{definition}

Constructors of a sort $\sort \in \indsort$ are meant to define the
domain of $\sort$, see Definition \ref{def:indsort}. The constant symbols
that are not constructors can be seen as existential variables
denoting arbitrary elements of a sort in $\indsort$ (notice however that $\const$ possibly contains constant symbols).
We assume that $\indsort$ contains a sort symbol $\nat$, with two constructors $0: \nat$ and $\su: \nat \rightarrow \nat$.

\newcommand{\n}{N}
\newcommand{\ii}{K}

\newcommand{\paramA}{A}
\newcommand{\paramB}{B}
\newcommand{\paramC}{C}

{\small
\begin{example}
Assume that we intend to reason on lists of elements of an arbitrary sort $\sort$.
Then $\S$ contains the sort symbols $\sort$ and $\lists$, where $\indsort = \{ \lists \}$.
The constructors are $\nil: \lists$ and $\cons: \sort \times \lists \rightarrow \lists$.
The set of parameters contains constant symbols of sorts $\sort$ or $\lists$ (denoting respectively elements and lists).
If $\paramA_1,\paramA_2$ are parameters of sort $\sort$, then $\cons(\paramA_1,\cons(\paramA_2,\nil))$ is a term of sort $\lists$.

Similarly, if one wants to reason on lists of natural numbers, then one should take $\indsort=\S=\{ \nat,\lists\}$. In this case, $\const = \{ \nil:\lists,\ \cons:\nat\times \lists \rightarrow \lists,\ 0: \nat,\ \su:\nat \rightarrow \nat\}$.
\end{example}
}

\newcommand{\indground}{$\indsort$-ground\xspace}

\newcommand{\dpred}{d}
\newcommand{\cpred}{c}

\newcommand{\appred}[2]{#1_{#2}}
\newcommand{\predef}[2]{\text{Def}^{#1}_{#2}}
\newcommand{\predeft}[2]{#1_{#2}\hspace*{-0.1cm}\downarrow_{\rules}}

Let $(\D_{\sort})_{\sort \in \indsort}$ be a family of disjoint sets of \emph{defined symbols} of sort $\sort$, disjoint from $\Sigma$, and
$\D \isdef \bigcup_{\sort \in \indsort} \D_{\sort}$.
An \emph{atom} is either an \emph{equation} of the form $t \iseq s$, where $t,s$ are terms of the same sort, or a \emph{defined atom}, of the form $\appred{\dpred}{t}$, where $\dpred \in \D_{\sort}$, for some $\sort \in \indsort$,  and $t \in \terms{\sort}$.
The arguments of the symbols in $\D$ are written as indices in order to distinguish them from predicate symbols that may occur in $\Sigma$ (such predicate symbols may be encoded as functions of profile $\vec{\sort} \rightarrow \bool$).
Formul{\ae} are built as usual on this set of atoms using the connectives $\vee,\wedge,\neg,\forall,\exists$. We assume for simplicity that all formul{\ae} are in Negation Normal Form (NNF).
A variable $x$ is \emph{free} in $\phi$ if it occurs in $\phi$, but not in the scope of the quantifier $\forall x$ or $\exists x$.
If $\phi$ has no free variables then $\phi$ is \emph{closed}.

\newcommand{\bform}{base formula\xspace}
\newcommand{\bforms}{base formul{\ae}\xspace}
\newcommand{\Bform}{Base formula\xspace}
\newcommand{\Bforms}{Base formul{\ae}\xspace}

\newcommand{\iform}{$\indsort$-formula\xspace}
\newcommand{\iforms}{$\indsort$-formul{\ae}\xspace}
\newcommand{\aniform}{an $\indsort$-formula\xspace}
\newcommand{\bterm}{base term\xspace}
\newcommand{\bterms}{base terms\xspace}

An \emph{interpretation} $I$ maps every sort $\sort$ to a set of elements $\intof{\sort}{I}$, every variable $x$ of sort $\sort$ to an element $\intof{x}{I} \in \intof{\sort}{I}$, every function symbol $f: \sort_1 \times \dots \times \sort_n \rightarrow \sort$ to a function $\intof{f}{I}$ from $\intof{\sort_1}{I} \times \dots \times \intof{\sort_n}{I}$ to $\intof{\sort}{I}$ and every defined symbol
$\dpred \in D_{\sort}$
to a subset of $\intof{\sort}{I}$.
The set $\bigcup_{\sort \in \S}  \intof{\sort}{I}$ is the \emph{domain} of $I$.
As usual, any interpretation $I$ can be extended to a function mapping
 every term $t$ of sort $\sort$
 to an element $\valof{t}{I} \in \intof{\sort}{I}$
 and every formula $\phi$ to a truth value $\valof{\phi}{I} \in \{ \truevalue, \falsevalue \}$.
We write $I \models \phi$ (and we say that $I$ \emph{validates} $\phi$) if $\valof{\forall \vec{x}\, \phi}{I} = \truevalue$, where $\vec{x}$ is the vector of free variables in $\phi$.
We assume, w.l.o.g., that the sets $\intof{\sort}{I}$ (for $\sort \in \S$) are disjoint.
Sets of formul{\ae} are interpreted as conjunctions.
 If $\phi$ and $\psi$ are two formul{\ae} or sets of formul{\ae}, we write $\phi \equiv_{I} \psi$ if either $I \models \phi$ and $I \models \psi$ or $I \not \models \phi$ and $I \not \models \psi$. We write $\phi \equiv \psi$ if $\phi \equiv_{I} \psi$ for all interpretations $I$.

\newcommand{\varrep}{variable replacement}
\newcommand{\imapping}{$\indsort$-mapping\xspace}
\newcommand{\imappings}{$\indsort$-mappings\xspace}

\newcommand{\imap}{\lambda}

We introduce two transformations operating on interpretations. The first one is simple: it only affects the value of some variables or constant symbols.
 If $I$ is an interpretation, $x_1,\dots,x_n$ are distinct variables or constant symbols of sort
  $\sort_1,\dots,\sort_n$ respectively and $v_1,\dots,v_n$ are
  elements of $\intof{\sort_1}{I},\dots,\intof{\sort_n}{I}$, then we denote
  by $I[v_1/x_1,\dots,v_n/x_n]$ the interpretation coinciding with
  $I$, except that for every $i = 1,\dots,n$, we have:
  $\intof{x_i}{I[v_1/x_1,\dots,v_n/x_n]} \isdef v_i.$

The second transformation is slightly more complex. The idea is to change the values of the elements of an inductive sort, without affecting the remaining part of the interpretation.
An {\em \imapping} for an interpretation $I$
is a function $\imap$ mapping every element $e$ in the domain of $I$ to an element of the same sort, that is the identity on every element occurring in a set $\intof{\sort}{I}$, where $\sort \not \in \indsort$.
Then $\imap(I)$ is the interpretation coinciding with $I$, except that for every symbol $f$ of a sort $\sort \not \in \indsort$, we have:
$\intof{f}{\imap(I)}(e_1,\dots,e_n) \isdef \intof{f}{I}(\imap(e_1),\dots,\imap(e_n))$.

\newcommand{\builtin}{\B}

\newcommand{\decomp}[1]{\Delta^{#1}}
\newcommand{\decompe}[2]{\Delta(#1\iseq#2)}

\newcommand{\classint}{{\mathfrak I}}
\newcommand{\classform}{{\mathfrak F}}

\newcommand{\nn}{M}
\newcommand{\rules}{{\frak R}}
\newcommand{\sdash}{\vdash_{\bot}}

In the following, we assume that all interpretations belong to a specific class $\classint$. This is useful to fix the semantics of some of the symbols, for instance one may assume that the interpretation of a sort $\integers$ is not arbitrary but rather equal to $\Z$.
Of course, $\classint$ is not arbitrary: the following definitions specify all the conditions that must be satisfied by the considered class of interpretations.
We start by the interpretation of the defined symbols.
As explained in the Introduction, the value of these symbols are to be
specified by convergent  systems of rewriting rules, satisfying some additional conditions defined as follows:

\begin{definition}
\label{def:rules}
Let $<$ be an ordering on defined symbols.
Let $\rules$ be an orthogonal system
of rules of the form $\dpred_{f(x_1,\dots,x_n)} \rightarrow \phi$,
where $\dpred$ is a defined symbol in $\sort$, $f$ is of profile $\sort_1 \times \dots \times \sort_n \rightarrow \sort$, and
$x_1,\dots,x_n$ are distinct variables of sorts $\sort_1,\dots,\sort_n$.
We assume that $\phi$ and $\rules$ satisfy the following conditions:
\begin{enumerate}
\item{The free variables of $\phi$ occur in $x_1,\dots,x_n$.}
\item{All \iterms occurring in $\phi$ belong to the set  $\{ x_1,\dots,x_n,f(x_1,\dots,x_n)\}$.}
\item{If $\phi$ contains a formula $\appred{\dpred'}{t}$ then either $\dpred' < \dpred$ and $t = f(x_1,\dots,x_n)$, or
    $t\in \{ x_1,\dots,x_n\}$.}
    \item{For every constructor $f$, $\rules$ contains a rule of the form $\dpred_{f(x_1,\ldots,x_n)} \rightarrow \phi$.}

\end{enumerate}
\end{definition}

It is clear from the conditions of Definition \ref{def:rules} that $\rules$ is convergent (the condition on the ordering ensures termination, and orthogonality ensures confluence).
We denote by $\predeft{\dpred}{t}$ the normal form of $\dpred_{t}$ w.r.t. $\rules$.
The following condition states that the interpretation of defined symbols
must correspond to the one specified by the rewrite system $\rules$, for every interpretation in $\classint$.

\newcommand{\rcomp}{$\rules$-compatible\xspace}

\begin{definition}
\label{def:rcomp}
An interpretation is {\em \rcomp} iff
for all sort symbols $\sort \in \indsort$, for all defined symbols $d \in \D_{\sort}$, for all function symbols $f: \sort_1 \times \dots \times \sort_n \rightarrow \sort$, we have $\appred{\dpred}{f(x_1,\dots,x_n)} \equiv_I \predeft{\dpred}{f(x_1,\dots,x_n)}$.
\end{definition}

The second condition that is required ensures that any equation between two constructor terms can be reduced to equations between variables:

\newcommand{\decompint}{$\iseq$-decomposable\xspace}
\newcommand{\eqdef}{\mathrm{eq}}
\newcommand{\noteqdef}{\overline{\mathrm{eq}}}

\begin{definition}
\label{def:decomp}
An interpretation is {\em \decompint} iff the following conditions hold:
\begin{enumerate}
\item{
For every $\sort \in \indsort$ and for every $f,g \in \Sigma_{\sort}$ of arity $n$ and $m$ respectively, there exists a formula $\decomp{(f,g)}$ built on $\vee,\wedge,\iseq$ and on $n+m$ distinct variables $x_1,\dots,x_n,y_1,\dots,y_m$
    such that $f(x_1,\dots,x_n)\iseq g(y_1,\dots,y_m) \equiv_I \decomp{(f,g)}$.\label{decomp:form}}
\item{For every $i\in[1,n]$ we have $\decomp{(f,g)} \models \bigvee_{k=1}^m x_i \iseq y_k$,
    and for every $j \in [1,m]$, we have $\decomp{(f,g)} \models \bigvee_{k=1}^n y_j \iseq x_k$.\label{decomp:eqargs}}
\end{enumerate}
   If $t = f(t_1,\dots,t_n)$ and $s = g(s_1,\dots,s_m)$ are two non-variable \iterms, we denote by $\decompe{t}{s}$
    the formula obtained from $\decomp{(f,g)}$ by replacing each variable $x_i$ ($1 \leq i \leq n$) by $t_i$ and each variable $y_j$ ($1 \leq j \leq m$) by $s_j$.
\end{definition}

{\small
\begin{example}
If, for instance, elements of a sort $\sort \in \indsort$ are interpreted as terms built on a set of free constructors, then
we have $\decomp{(f,g)}\iseq \falseform$ if $f \not = g$ and
$\decomp{(f,f)} \isdef x_1\iseq y_1 \wedge \dots \wedge x_n\iseq y_n$ (where $n$ denotes the arity of $f$).
Indeed, in this case, we have $f(x_1,\ldots,x_n) \iseq f(y_1,\ldots,y_n) \equiv (x_1\iseq y_1 \wedge \dots \wedge x_n\iseq y_n)$.
If, on the other hand, $g$ is intended to denote a commutative binary function then we should have:
$\decomp{(g,g)} = (x_1\iseq y_1 \wedge x_2\iseq y_2) \vee (x_1\iseq y_2 \wedge x_2\iseq y_1)$.
The variables $x_i$ and $y_j$ are those introduced in Definition \ref{def:decomp}.
\end{example}
}

The third condition
ensures that the interpretation of every inductive sort is minimal (w.r.t. to set inclusion).

\newcommand{\inductive}{$\indsort$-inductive\xspace}

\begin{definition}
\label{def:indsort}
An interpretation is {\em \inductive} iff
for every $\sort \in \S$,
and for every element $u \in \intof{\sort}{I}$,
there exists a constructor term
$t$
such that $u = \valof{t}{I}$.
\end{definition}

Notice that, by definition, a constructor term contains no variable of a sort in $\indsort$. For instance, every element in $\intof{\nat}{I}$ should be equal to a ground term $\su^k(0)$, for some $k \in {\Bbb N}$. If $\lists$ denotes the sort of the lists built on elements of a sort $\sort \not \in \indsort$, then any element of $\intof{\lists}{I}$ must be equal to a term of the form $\cons(x_1,\cons(x_2,\ldots,\cons(x_n,\nil)\ldots))$, where $x_1,\ldots,x_n$ are variables of sort $\sort$.
This condition implies in particular that for every $\sort \not \in \indsort$ and for every element $v \in \intof{\sort}{I}$, there exists a variable $x$ such that $\intof{x}{I}=v$ (this is obviously not restrictive, since the variables may be interpreted arbitrarily).

\newcommand{\nice}{schematizable\xspace}

The next definition summarizes all the conditions that are imposed:

\begin{definition}
\label{def:int}
\label{def:nice}
A class of interpretations $\classint$ is {\em \nice} iff all interpretations $I \in \classint$ satisfy the following properties:
\begin{enumerate}
\item{$I$ is \rcomp.\label{int:def}}

\item{$I$ is \decompint.\label{int:decomp}}

\item{$I$ is \inductive.\label{int:ind}}

\item{For all variables $v$ of a sort $\sort$ and for all elements $e \in \intof{\sort}{I}$, $I[e/v] \in \classint$.\label{int:varrep}}

\item{For all \imappings $\imap$, $\imap(I) \in \classint$. \label{int:stable}}

\end{enumerate}
A  formula $\phi$ is \emph{$\classint$-satisfiable} iff $\phi$ has a model in $\classint$.

\end{definition}

\newcommand{\niceform}{admissible\xspace}

From now on we focus on testing $\classint$-satisfiability for a \nice class of interpretations. Before that we impose some restrictions on the formul{\ae} to be tested.
As we shall see, these conditions
will be useful mainly to ensure that the proof procedure presented in Section \ref{sect:proof} only generates a finite number of distinct formul{\ae}, up to a renaming of the parameters.
This property is essential for the proof of termination, although it is not a sufficient condition.

\begin{definition}
\label{def:niceform}
A class of formul{\ae} $\classform$ is {\em \niceform} if all formul{\ae} $\phi \in \classform$ satisfy the following properties:
\begin{enumerate}
\item{For all parameters $\paramA,\paramB$, $\replaceby{\phi}{\paramA}{\paramB} \in \classform$.\label{niceform:subst}}
\item{$\phi$ contains no constructor and no variable of a sort in $\indsort$.\label{niceform:noconst}}

\item{For every subformula $\psi$ of $\phi$, if $\psi$ is not a disjunction, a conjunction, or a defined atom,
then $\psi$ contains no defined symbol and no pairs of distinct parameters.\label{niceform:param}}
\item{For every defined symbol $\dpred$ occurring in $\phi$ and for every rule $\appred{\dpred}{t} \rightarrow \phi$ in $\rules$, the formula obtained from $\phi$ by replacing each $\indsort$-term by an arbitrary parameter is in $\classform$.\label{niceform:rules}}
\end{enumerate}
A formula occurring in $\classform$ is a \emph{schema}.
It is a {\em \bform} iff it contains no defined symbol,
and no equation between parameters.
\end{definition}

The conditions in Definition \ref{def:niceform} ensure that the formul{\ae} in $\classform$ are boolean combinations (built on $\vee$,$\wedge$) of \bforms containing at most one parameter, of
defined atoms and of
equations and disequations between parameters.
The definition of \bforms in Definition \ref{def:niceform}
ensures that the truth values of \bforms do not depend
on the interpretation of the parameters, but only on the \emph{relation} between them.
\Bforms can contain parameters, but they can only occur as arguments of function symbols, whose images must be of a non-inductive sort. The only way of specifying properties of the parameters themselves (and not of the terms built on them) is by using the rewrite rules in $\rules$.
As we shall see, this property is essential for proving the soundness of the loop detection rule that ensures termination of our proof procedure. Similarly, no quantification over variables of an inductive sort is allowed.

In the following, $\classint$ denotes a \nice class of interpretations and
$\classform$ denotes an \niceform class of formul{\ae}.
The goal of the paper is to prove that if $\classint$-satisfiability is decidable (resp. semi-decidable)
for \bforms in $\classform$ then it must be so for
all formul{\ae} in $\classform$.
We give examples of classes of formul{\ae} satisfying the previous conditions:

\newcommand{\classintprop}{\classint_0}
\newcommand{\classintint}{\classint_{\Bbb Z}}
\newcommand{\classintfol}{\classint_1}

\newcommand{\classformprop}{\classform_0}
\newcommand{\classformint}{\classform_{\Bbb Z}}
\newcommand{\classformfol}{\classform_1}

{\small
\begin{example}
\label{ex:propsch}
Assume that $\Sigma$ only contains $0$, $\su$ and symbols of profile $\nat \rightarrow \bool$.
Let $\classintprop$ be the class of all \rcomp interpretations on this language with the usual interpretation of $\nat$, $0$ and $\su$, and let $\classformprop$ be the set of all quantifier-free formul{\ae} containing no occurrence of $0$ and $\su$.
Clearly, $\classintprop$ is \nice and $\classformprop$ is \niceform.
The formul{\ae} in $\classformprop$ denote schemata of propositional formul{\ae}.
For instance the schema $p_0 \wedge \neg p_{\n} \wedge \bigwedge_{\ii=0}^{\n-1} (\neg p_\ii \vee p_{\su(\ii)})$ is specified by the formul{\ae}:
$p(0) \wedge \neg p(\n) \wedge \dpred_\n$, where $\dpred$ is defined by the rules
$\dpred_0 \rightarrow \trueform$ and
$\dpred_{\su(\ii)} \rightarrow \dpred_\ii \wedge (\neg p(\ii) \vee p(\su(\ii)))$.
$\classformprop$ is  equivalent to the class of \emph{regular schemata} in \cite{ACP10}.
\end{example}
}

\newcommand{\nameclassint}{Proposition \ref{prop:classintprop}}

{\small
\begin{example}
\label{ex:arithsch}
Let $\S = \{ \nat, \integers \}$ and $\indsort = \{ \nat \}$.
Assume that $\Sigma$ contains the symbols $0$ and $\su$, constant symbols of sort $\integers$, function symbols of profile $\nat \rightarrow \integers$
and all the symbols of Presburger arithmetic.
Let $\classintint$ be the class of all \rcomp interpretations
such that the interpretations of $\nat, \integers, 0,\su,+,\leq,\dots$ are the usual ones.
Let $\classformint$ be the set of all formul{\ae} built on this language, containing no occurrence of $0$, $\su$, and satisfying Condition
\ref{niceform:param} in Definition \ref{def:niceform}.
It can be easily checked that $\classintint$ is \nice and that $\classformint$ is \niceform.
Formul{\ae} in $\classformint$ denote schemata of Presburger formul{\ae} (the \bforms in $\classformint$ are formul{\ae} of Presburger arithmetic). For instance $\bigvee_{\ii=0}^\n a(\ii) > 0$ is denoted by $\dpred_\nn$, with the rules
$\dpred_0 \rightarrow (a(0)>0)$ and
$\dpred_{\su(\ii)} \rightarrow \dpred_\ii \vee a(\su(\ii)) > 0$.
Note however, that schemata containing atoms with several distinct terms of sort $\nat$, such as $\bigwedge_{\ii=0}^\n a(\ii)\iseq a(\su(\ii))$ cannot occur in $\classformint$.
It is also important to remark that the sort $\integers$ \emph{must} be distinct from the sort of the indices $\nat$ (terms of the form $\dpred_{a(\ii)}$ are \emph{not} allowed).
\end{example}
}

The class $\classformint$ is not comparable to the class of SMT-schemata in \cite{AP11a} (the latter class may contain formul{\ae} of the previous form, at the cost of additional restrictions on the considered theory).
Let $\classintfol$ and $\classformfol$ be the sets of interpretations and formul{\ae} fulfilling the conditions of Definitions \ref{def:nice} and \ref{def:niceform}.
The following proposition is easy to establish ($\classformprop$ and $\classformint$ are defined in Examples \ref{ex:propsch} and \ref{ex:arithsch}):

\begin{proposition}
\label{prop:classintprop}
$\classintprop$-satisfiability (resp. $\classintint$-satisfiability)
is decidable for
\bforms in $\classformprop$ (resp. $\classformint$), and
$\classintfol$-satisfiability is semi-decidable for \bforms in $\classformfol$.
\end{proposition}

\newcommand{\prclassint}{
\begin{proof}
Since the \bforms contain no equations between elements of sort $\nat$, we can assume that all parameters are mapped to distinct natural numbers (it is clear that this operation preserves satisfiability). Then any formula in $\classformprop$ (resp. $\classformint$, resp. $\classformfol$)
is essentially equivalent to a propositional formula (resp. to a formula of Presburger arithmetic, resp. to a first-order formula).
\end{proof}
}

\addproof{\protect\prclassint}{\nameclassint}

Before describing the proof procedure for testing the satisfiability of schemata, we provide a simple example of an application. It is only intended to give a taste of what can be expressed in our logic, and of which properties are outside its scope (see also the examples in the Introduction, that can be easily encoded).

\newcommand{\constant}[2]{A^{#1,#2}}
\newcommand{\constantb}[2]{E^{#1,#2}}
\newcommand{\constantr}[2]{\text{Cr}^{#1,#2}}
\newcommand{\constantl}[2]{\text{Cl}^{#1,#2}}
\newcommand{\switch}[2]{\text{Sw}^{#1,#2}}

\newcommand{\sortedA}[1]{\text{Sort}^{#1}}
\newcommand{\altern}[3]{\text{Alt}^{#1,#2,#3}}
\newcommand{\elem}{{\tt elem}}

\newcommand{\sorted}[1]{S^{#1}}
\newcommand{\maparray}[3]{\text{Map}^{#1,#2,#3}}
\newcommand{\construct}[3]{p^{#1,#2,#3}}
\newcommand{\treesort}{{\tt DAG}}
\newcommand{\emptytree}{\bot}
\newcommand{\btree}{\mathrm{c}}
\newcommand{\TrA}{\delta}
\newcommand{\TrB}{\delta'}
\newcommand{\TrC}{\delta''}
\newcommand{\eqt}{E}

{\small
\begin{example}
A (binary) DAG $\TrA$ labeled by elements of type $\elem$ can be denoted by a function symbol
$\TrA: \treesort \rightarrow \elem$, where the signature contains two constructors of sort $\treesort$: a constant symbol $\emptytree$ (denoting the empty DAG),
and a $3$-ary symbol $\btree(n,l,r)$, where $l$ and $r$ denote the left and right children respectively and $n$ denotes the current node\footnote{This extra-argument is necessary to ensure that distinct nodes can have the same children.}.
Various properties can be expressed in our logic, for instance the following defined symbol $\constant{\TrA}{p}_x$ expresses the fact that all the elements occurring in a DAG $\TrA$ satisfies some property $p$.

\[\constant{\TrA}{p}_{\emptytree} \rightarrow \trueform \qquad \quad
\constant{\TrA}{p}_{\btree(n,l,r)} \rightarrow \constant{\TrA}{p}_{l} \wedge \constant{\TrA}{p}_{r} \wedge p(\TrA(\btree(n,l,r)))\]

Obviously this can be generalized to any set of regular positions: for instance, we can state that there exists a path from the root to a leaf in the DAG on which all the element satisfy $p$:

\[\constantb{\TrA}{p}_{\emptytree} \rightarrow \trueform \qquad \quad \constantb{\TrA}{p}_{\btree(n,l,r)} \rightarrow (\constantb{\TrA}{p}_{l} \vee \constantb{\TrA}{p}_{r}) \wedge p(\TrA(\btree(n,l,r)))\]

$\TrA$ and $p$ are meta-variables: $\TrA$ must be replaced by a function symbol of profile $\treesort \rightarrow \elem$ and
$p$ can be replaced by any property of elements of sort $\elem$ (provided it is expressible in the base language e.g. first-order logic).
For instance, we can express the fact that all the elements of $\TrA$ are equal to some fixed value, or that all the elements of $\TrA$ are even.
We can check that the following formula is valid:
$(\forall x, p(x) \Rightarrow q(x)) \Rightarrow (\constantb{\TrA}{p} \Rightarrow \constantb{\TrA}{q})$.
However, the converse \emph{cannot} be expressed in our setting, because it would involve a quantification over an element of type $\treesort$ which is forbidden by Condition \ref{niceform:noconst} in Definition \ref{def:niceform}.
The formula $\constant{\TrA}{p} \wedge \neg \constant{\TrA}{q} \wedge \neg \constant{\TrA}{\neg q}$ is satisfiable on the interpretations
whose domain contains two elements $e_1$, $e_2$ such that $p(e_1), p(e_2)$, $\neg q(e_1)$, and $q(e_2)$ hold (but for instance it is unsatisfiable if $p(x) \equiv (x \iseq 0)$).
We can express the fact that two DAGs $\TrA$ and $\TrB$ share an element: $\exists x,\, \forall y,\, (p(y) \Leftrightarrow x=y) \wedge \neg \constant{\TrA}{\neg p} \wedge \neg \constant{\TrB}{\neg p}$.
 We can also define a symbol $\maparray{\TrA}{\TrB}{f}$ stating that $\TrB$ is obtained from $\TrA$ by applying some function $f$ on every element of $\TrA$:

\[\begin{tabular}{c}
$\maparray{\TrA}{\TrB}{f}_{\emptytree} \rightarrow \trueform$ \\
$\maparray{\TrA}{\TrB}{f}_{\btree(n,l,r)} \rightarrow \maparray{\TrA}{\TrB}{f}_l \wedge \maparray{\TrA}{\TrB}{f}_r \wedge  \TrB(\btree(n,l,r)) = f(\TrA(\btree(n,l,r)))$ \\
\end{tabular}
\]
Then, we can check, for instance, that if all the elements of $\TrA$ are even and if $f$ is the successor function, then all the elements of $\TrB$ must be odd:
$$(\even(0) \wedge (\forall x,\, \even(\su(x)) \Leftrightarrow \neg \even(x)) \wedge \constant{\TrA}{\even}_\paramA) \wedge \maparray{\TrA}{\TrB}{\su} \Rightarrow \constant{\TrB}{\neg \even}_\paramA$$
We are not able, however, to express transformations affecting the \emph{shape} of the DAG (e.g. switching all the right and left subgraphs) because this would require to use non-monadic defined symbols.

$\altern{\TrA}{p}{q}$ expresses the fact that all the elements at even positions satisfy $p$ and that the elements at odd positions satisfy $q$:

\[\altern{\TrA}{p}{q}_{\emptytree} \rightarrow \trueform \quad
\altern{\TrA}{p}{q}_{\btree(n,l,r)} \rightarrow \altern{\TrA}{q}{p}_{l} \wedge \altern{\TrA}{q}{p}_{r} \wedge p(\TrA(\btree(n,l,r)))\]

Our procedure can be used to verify that  $\altern{\TrA}{p}{q}_\paramA \Rightarrow \constant{\TrA}{p\vee q}_\paramA$.
The following defined symbol $\construct{\TrA}{\TrB}{\TrC}$ states that a DAG $\TrC$ is constructed by taking elements from $\TrA$ and $\TrB$ alternatively:

\[
\begin{tabular}{c}
$\construct{\TrA}{\TrB}{\TrC}_{\emptytree} \rightarrow \trueform$ \\
$\construct{\TrA}{\TrB}{\TrC}_{\btree(n,l,r)} \rightarrow \construct{\TrB}{\TrA}{\TrC}_{l} \wedge \construct{\TrB}{\TrA}{\TrC}_{r} \wedge \TrC(\btree(n,l,r))=\TrA(\btree(n,l,r))$ \\
\end{tabular}
\]
We can check that if the elements of $\TrA$ and $\TrB$ satisfy Properties $p$ and $q$ respectively, then the elements in $\TrC$ satisfy $p$ and $q$ alternatively:
$(\construct{\TrA}{\TrB}{\TrC}_\paramA \wedge \constant{\TrA}{p}_\paramA \wedge \constant{\TrB}{q}_\paramA) \Rightarrow \altern{\TrC}{p}{q}_\paramA$.

Notice that, in this example, the subgraphs can share elements. Thus it is not possible in general to reason independently on each branch (in the style of automata-based approaches): one has to reason \emph{simultaneously} on the whole DAG.
Other data structures such as arrays or lists can be handled in a similar way.
An example of property that \emph{cannot} be expressed is sortedness. Indeed, it would be stated as follows:
\[\sortedA{\TrA}_{\btree(n,l,r)} \rightarrow \sortedA{\TrA}_{l} \wedge \sortedA{A}_{r} \wedge \TrA(\btree(n,l,r)) \geq \TrA_l  \wedge \TrA(\btree(n,l,r)) \geq \TrA_r\]

However, the atom $\TrA(\btree(n,l,r)) \geq \TrA_l$
is \emph{not} allowed in our setting: since it contains several parameters, it contradicts Condition \ref{niceform:param} in Definition \ref{def:niceform}.
\end{example}
}

\section{Proof Procedure}

\newcommand{\repr}{\mathrm{rep}}
\newcommand{\depth}{\text{\it depth}}
\newcommand{\mint}[2]{\mathrm{depth}_{#2}(#1)}
\newcommand{\minterm}[2]{\repr_{#2}(#1)}

\newcommand{\parent}[1]{\geq_{#1}}
\newcommand{\child}[1]{\leq_{#1}}

\newcommand{\solved}{solved\xspace}
\newcommand{\renaming}{renaming\xspace}
\newcommand{\ren}{\rho}
\newcommand{\moregen}{\sqsupseteq}

\newcommand{\labelsof}[2]{#2(#1)}
\newcommand{\labelsofbis}[2]{#2'(#1)}
\newcommand{\rulefont}[1]{{\sc #1}}
\newcommand{\startrule}{\rulefont{Start}\xspace}
\newcommand{\veerule}{$\vee$-\rulefont{Decomposition}\xspace}
\newcommand{\wedgerule}{$\wedge$-\rulefont{Decomposition}\xspace}
\newcommand{\clashrule}{\rulefont{Closure}\xspace}
\newcommand{\eclashrule}{\rulefont{$\iseq$-Closure}\xspace}
\newcommand{\inductionrule}{\rulefont{Unfolding}\xspace}
\newcommand{\eqrule}{$\iseq$-\rulefont{Decomposition}\xspace}
\newcommand{\diseqrule}{$\not \iseq$-\rulefont{Decomposition}\xspace}
\newcommand{\depthrule}{$\maxs$-\rulefont{Rule}\xspace}
\newcommand{\depthclash}{$\less$-\rulefont{Clash}\xspace}
\newcommand{\simprule}{$\n$-\rulefont{Ssimplification}\xspace}
\newcommand{\lessclashrule}{$\less$-\rulefont{Separation}\xspace}
\newcommand{\lessrule}{$\less$-\rulefont{Decomposition}\xspace}
\newcommand{\strictrule}{\rulefont{Strictness}\xspace}
\newcommand{\reprule}{\rulefont{Replacement}\xspace}
\newcommand{\delrule}{\rulefont{\Deletion}\xspace}
\newcommand{\seprule}{\rulefont{Separation}\xspace}
\newcommand{\purrule}{\rulefont{\Purification}\xspace}
\newcommand{\looprule}{\rulefont{Loop}\xspace}
\newcommand{\explrule}{\rulefont{Explosion}\xspace}
\newcommand{\depthexplrule}{$\n$-\rulefont{Explosion}\xspace}
\newcommand{\inductionbotrule}{\rulefont{$\D$-Closure}\xspace}
\newcommand{\Startrule}{\rulefont{Start}\xspace}
\newcommand{\Veerule}{$\vee$-\rulefont{Decomposition}\xspace}
\newcommand{\Wedgerule}{$\wedge$-\rulefont{Decomposition}\xspace}
\newcommand{\Clashrule}{\rulefont{Closure}\xspace}
\newcommand{\Eclashrule}{\rulefont{$\iseq$-Closure}\xspace}
\newcommand{\Nclashrule}{$\n$-\rulefont{Closure}\xspace}
\newcommand{\nclashrule}{$\n$-\rulefont{Closure}\xspace}
\newcommand{\Inductionrule}{\rulefont{Unfolding}\xspace}
\newcommand{\Inductionbotrule}{\rulefont{$\D$-Closure}\xspace}
\newcommand{\Eqrule}{$\iseq$-\rulefont{Decomposition}\xspace}
\newcommand{\Diseqrule}{$\not \iseq$-\rulefont{Decomposition}\xspace}
\newcommand{\Depthrule}{$\maxs$-\rulefont{Rule}\xspace}
\newcommand{\Depthclash}{$\n$-\rulefont{Clash}\xspace}
\newcommand{\Lessrule}{$\less$-\rulefont{Decomposition}\xspace}
\newcommand{\Simprule}{$\n$-\rulefont{Simplification}\xspace}
\newcommand{\Lessclashrule}{$\less$-\rulefont{Separation}\xspace}
\newcommand{\Strictrule}{\rulefont{Strictness}\xspace}
\newcommand{\Reprule}{\rulefont{Replacement}\xspace}
\newcommand{\Delrule}{\rulefont{Deletion}\xspace}
\newcommand{\Seprule}{\rulefont{Separation}\xspace}
\newcommand{\Purrule}{\rulefont{Purification}\xspace}
\newcommand{\Looprule}{\rulefont{Loop}\xspace}
\newcommand{\Explrule}{\rulefont{Explosion}\xspace}
\newcommand{\Depthexplrule}{$\n$-\rulefont{Explosion}\xspace}
\newcommand{\invertiblerules}{\veerule, \wedgerule, \clashrule, \nclashrule, \eclashrule, \inductionrule,
\eqrule, \reprule, \seprule,  \strictrule,
\lessclashrule and \lessrule}
\newcommand{\satrules}{
\explrule and \depthexplrule}
 \newcommand{\layer}{layer\xspace}
 \newcommand{\layerform}{layer formula\xspace}
 \newcommand{\layerforms}{layer formul{\ae}\xspace}
 \newcommand{\layers}{layers\xspace}
\newcommand{\decomprules}{Decomposition Rules}
\newcommand{\equalityrules}{Equality Rules}
\newcommand{\depthrules}{Depth Rules}
\newcommand{\prooftree}{proof tree\xspace}
\newcommand{\prooftrees}{proof trees\xspace}
\newcommand{\Prooftree}{Proof \Tree\xspace}
\newcommand{\Prooftrees}{Proof \Trees\xspace}
\newcommand{\inductionrules}{Unfolding Rule}
\newcommand{\proofderivation}{derivation\xspace}
\newcommand{\appcond}[1]{\quad {\em #1}}
\newcommand{\depictstart}{
\text{\Startrule:}
\begin{tabular}{cc}
\begin{tabular}{c}
\mbox{} \\
\hline
$\phi, \max(\{ \depth(A_i) \mid i \in [1,n] \}) \iseq \n$
\end{tabular} &
\begin{tabular}{l}
\appcond{Where $\phi$ denotes the formula at hand} \\
\appcond{$\paramA_1,\dots,\paramA_n$ are the parameters in $\phi$}
\end{tabular}
\end{tabular}
}
\newcommand{\depictdec}{
\begin{tabular}{rcrcr}
\Veerule: & \begin{tabular}{c}
$\phi \vee \psi$ \\
\hline
$\phi$ \quad \vline \quad $\psi$ \\
\end{tabular}
\quad
\Wedgerule: & \begin{tabular}{c}
$\phi \wedge \psi$ \\
\hline
$\phi,\psi$ \\
\end{tabular} \quad &
\begin{tabular}{ll}
\appcond{If $\phi \wedge \psi$ is} \\
\appcond{not a \bform}
\end{tabular} \\
\end{tabular}
}
\newcommand{\depictclos}{
\begin{tabular}{rc}
\Clashrule: & \begin{tabular}{c}
$\neg \phi, \phi$ \\
\hline
$\falseform$ \\
\end{tabular} \\
\end{tabular}
}
\newcommand{\depictnclos}{
\begin{tabular}{rc}
\Nclashrule: & \begin{tabular}{c}
$0 \iseq \su(t)$ \\
\hline
$\falseform$ \\
\end{tabular} \\
\end{tabular}
}
\newcommand{\depictinduction}{
\text{\Inductionrule:}\quad
\begin{tabular}{c}
$\appred{\dpred}{\paramA}, \paramA \iseq f(\vec{\paramB})$ \\
\hline
$\psi$
\end{tabular} \quad
\begin{tabular}{c}
$\neg \appred{\dpred}{\paramA}, \paramA \iseq f(\vec{\paramB})$ \\
\hline
$\NNF(\neg \psi)$
\end{tabular}
\appcond{$\psi = \replaceby{\predeft{\dpred}{f(\vec{\paramB})}}{f(\vec{\paramB})}{\paramA}, \paramA\iseq f(\vec{\paramB})$}
}

\newcommand{\depicteq}{
\[
\text{\Eqrule:} \begin{tabular}{c}
$\paramA \iseq f(\vec{\paramB}), \paramA \iseq g(\vec{\paramC})$ \\
\hline
$\decompe{f(\vec{\paramB})}{g(\vec{\paramC})}, \paramA \iseq f(\vec{\paramB})$ \\
\end{tabular}
\]
}
\newcommand{\depictdiseq}{
\[
\begin{tabular}{ll}
\text{\Diseqrule:} \begin{tabular}{c}
$\paramA \not \iseq \paramB, \paramA \iseq f(\vec{\paramB}), \paramB \iseq g(\vec{\paramC})$ \\
\hline
$\psi, \paramA \not \iseq \paramB, \paramA \iseq f(\vec{\paramB}), \paramB \iseq g(\vec{\paramC})$ \\
\end{tabular}
\appcond{If $\psi = \NNF(\neg \decompe{f(\vec{\paramB})}{g(\vec{\paramC})})$\footnote{See Definition \ref{def:decomp} for the definition of $\decompe{t}{s}$} }
\end{tabular}
\]
}
\newcommand{\depicteqdiseq}{
\[
\begin{tabular}{cc}
\text{\Eqrule:} \begin{tabular}{c}
$\paramA \iseq f(\vec{\paramB}), \paramA \iseq g(\vec{\paramC})$ \\
\hline
$\psi, \paramA \iseq f(\vec{\paramB})$ \\
\end{tabular} \quad & \quad
\begin{tabular}{c}
$\paramA \not \iseq \paramB, \paramA \iseq f(\vec{\paramB}), \paramB \iseq g(\vec{\paramC})$ \\
\hline
$\NNF(\neg \psi), \paramA \not \iseq \paramB, \paramA \iseq f(\vec{\paramB}), \paramB \iseq g(\vec{\paramC})$ \\
\end{tabular} \\
\multicolumn{2}{l}{\appcond{Where  $\psi = \decompe{f(\vec{\paramB})}{g(\vec{\paramC})}$\footnote{See Definition \ref{def:decomp} for the definition of $\decompe{t}{s}$}}}
\end{tabular}
 \]
}
\newcommand{\depicteclash}{
\text{\Eclashrule:} \begin{tabular}{c}
$A \not \iseq A$ \\
\hline
$\falseform$ \\
\end{tabular}
}
\newcommand{\depictsep}{
\text{\Seprule:} \begin{tabular}{c} \\
\hline
$\paramA\iseq \paramB \vee \paramA \not \iseq \paramB$
\end{tabular}
}
\newcommand{\depictrep}{
\[
\begin{tabular}{ll}
\text{\Reprule:}\begin{tabular}{c}
$\phi, \paramA\iseq \paramB$ \\
\hline
$\replaceby{\phi}{\paramA}{\paramB}, \paramA \iseq \paramB$
\end{tabular}
&
\appcond{
If $\paramA$ and $\paramB$ are two
parameters and $\paramA$ occurs in $\phi$
}
\end{tabular}
\]
}
\newcommand{\depictstrict}{
\text{\Strictrule:}
\begin{tabular}{c}
$\depth(\paramA) \lesseq \n$ \\
\hline
$\depth(\paramA)\iseq \n \vee \depth(\paramA) \less \n$
\end{tabular}
}
\newcommand{\depictless}{
\text{\Lessrule:}
\begin{tabular}{c}
$t \less \su(\n)$ \\
\hline
$t \lesseq \n$ \\
\end{tabular}
}
\newcommand{\depictdepth}{
\text{\Depthrule:}
\begin{tabular}{c}
$\maxs(E) \iseq \n$ \\
\hline
$\bigwedge_{t\in E} t \lesseq \n, \bigvee_{t \in E} t\iseq \n$ \\
\end{tabular}
}
\newcommand{\depictlessclash}{
\text{\Lessclashrule:}
\begin{tabular}{c}
$\depth(\paramA) \less \n, \depth(\paramB) \iseq \n$ \\
\hline
$\depth(\paramA) \less \n, \depth(\paramB) \iseq \n, \paramA \not \iseq \paramB$
\end{tabular}
}

\newcommand{\depictexpl}{
\text{\Explrule:}
\begin{tabular}{c}
$\depth(\paramB) \iseq \su(t)$ \\
\hline
$\bigvee_{i \in [1,n]} \maxs(E_i)\iseq t \wedge \paramB \iseq t_i$
\end{tabular}

\begin{quotation}
\appcond{If $t_i$ are terms of the form $f_i(\vec{\paramA_i})$, such that $f_1,\dots,f_n$ are all the function symbols of the same sort as $\paramB$, and the $\vec{\paramA_i}$'s are vectors of pairwise distinct, fresh, constant symbols of the appropriate sort, and $E_i$ is the set of terms $\depth(\paramC)$, where $\paramC$ is a component of $\vec{\paramA_i}$ of a sort in $\indsort$.}
\end{quotation}
}
\newcommand{\depictdepthexpl}{
\text{\Depthexplrule:}
\begin{tabular}{c}
$\Phi$ \\
\hline
$\replaceby{\Phi}{\n}{\su(0)}$ \quad \vline \quad $\replaceby{\Phi}{\n}{\su(\n)}$ \\
\end{tabular}

\begin{quotation}
\appcond{If no other rule applies and $\n$ occurs in $\Phi$. Notice that in contrast with the previous rules, $\Phi$ must denote the whole label (not a subset of it)}
\end{quotation}
}
\newcommand{\depictloop}{
\[
\begin{tabular}{ll}
\text{\Looprule:}
\begin{tabular}{c}
$\Phi$ \\
\hline
$\falseform$
\end{tabular}\quad\quad\quad
&
\begin{tabular}{l}
\appcond{If there exists in the same branch a (non leaf) \layer labeled by} \\
\appcond{a set of formul{\ae} $\Psi$ such that $\noneq{\Phi} \moregen \noneq{\Psi}$}
\end{tabular}
\end{tabular}
\]
}
\newcommand{\commentdecomp}{
The decomposition and closure rules are standard. However, we do \emph{not} use them to test the satisfiability of the formula, but only to decompose it into a conjunction of defined atoms, equational literals and \bforms.
This is always feasible, thanks to the particular properties of formul{\ae} in $\classform$ (see Definition \ref{def:niceform}).
Notice that the separation rule has no premises. The only requirement is that $\paramA$ and $\paramB$ occur in the considered branch.
}
\newcommand{\commentinduction}{
\inductionrule replaces a defined atom $\appred{\dpred}{\paramA}$ by its definition according to the rules in $\rules$. This is possible only when the head symbol and arguments of the term represented by $\paramA$ are known.}
\newcommand{\commenteq}{\Eqrule decomposes equalities, using the specific properties of \decompint interpretations: if a node contains two equations $\paramA \iseq t$ and $\paramA \iseq s$ then the formula $\decompe{t}{s}$ necessarily holds. \Diseqrule performs a similar task for inequalities.
}
\newcommand{\commentstart}{
\Startrule is only applied once, in order to create the root node of the tree. The label of this node contains the formula at hand together with an additional formula stating that the max of the depth of the constructor terms represented by the parameters must equal to some natural number $\n$.
}
\newcommand{\commentexpl}{
The Explosion rules instantiate the parameters, which is done by adding equations of the form $\paramA \iseq f(\vec{\paramB})$, where $\vec{\paramB}$ is a vector of fresh parameters.

\Explrule instantiates the parameters distinct from $\n$.
We choose to instantiate only the parameters representing terms of maximal depth, and only after $\n$ has been instantiated.
Thus we instantiate a parameter $\paramB$ only if there exists an atom of the form $\depth(\paramB) \iseq t$, where $t$ is of the form $\su(s)$, for some $s \in \{ 0, \n \}$. \Explrule enables further applications of \inductionrule, which in turn may introduce new complex formul{\ae} into the nodes (by unfolding the defined symbols according to the rules in $\rules$).

\Depthexplrule instantiates the parameter $\n$. Since the depth of the terms of a sort in $\indsort$ is at least $1$  and since $\n$ is intended to denote the maximal depth of the parameters, $\n$ cannot be $0$, thus it is instantiated either by $\su(0)$ or by $\su(\n)$. Unlike the other parameters, direct replacement is performed.
This rule is applied with the lowest priority.
Hence, when the rule is applied, all parameters of a depth strictly greater than $\n$ must have been instantiated. By replacing $\n$ by a term of the form $\su(t)$, the rule will permit to instantiate
the parameters of depth $\n-1$.
This strategy ensures that the parameters will be instantiated in decreasing order w.r.t. the depth of the terms they represent.
}
\newcommand{\commentdepth}{Several rules are introduced to reason on the depth of the terms represented by the parameters.
The principle is to separate the parameters representing terms of a depth exactly equal to $\n$ from those whose depth is strictly less than $\n$ (so that only the former ones may be instantiated).
By definition of \startrule, the initial node must contain an equation $\depth(\paramA) \lesseq \n$ for each parameter $\paramA \not = \n$.
\Strictrule expands this inequality by using the equivalence $x\lesseq y \Leftrightarrow (x \less y \vee x \iseq y)$. Then \veerule will apply, yielding either $x \less y$ or $x \iseq y$.
\Lessrule gets rid of strict equalities of the form $\depth(\paramA) \less \su(t)$ that are introduced by
\depthexplrule.
}
\newcommand{\commentloop}{\Looprule is intended to detect cycles and prune the corresponding branches, by closing the nodes that are subsumed by a previous one.
It only applies on some particular nodes, that are irreducible w.r.t. all rules, except (possibly) \depthexplrule. We shall call any such node a {\em \layer.}
This rule can be viewed as an application of the induction principle.
If
$\Phi \moregen \Psi$ then it is clear that $\Psi$ is a logical consequence of $\Phi$, up to a renaming of parameters.
Thus, if some open node exists below a node labeled by $\Phi$,
 some other open node must exist also below a node labeled by $\Psi$,
 hence the node corresponding to
 $\Phi$ may be closed without threatening soundness (a satisfiable branch is closed, but global satisfiability is preserved). Since $\Psi$ is a \layer, the parameter $\n$ must be instantiated at least once between the two nodes, which ensures that the reasoning is well-founded and that there exists at least one  open node outside the branch of $\Phi$.

At first glance, it may seem odd to remove equations from $\Phi$ and $\Psi$ before testing for subsumption (see the application condition of \looprule).
  Indeed, it is clear that this operation does \emph{not} preserve satisfiability in general.
For instance, the formula $p(\paramA) \wedge \neg p(\paramB) \wedge \appred{\dpred}{\paramB} \wedge \paramA \iseq 0$ is unsatisfiable
if $\dpred$ is defined by the rules:
$\appred{\dpred}{0} \rightarrow \trueform$ and $\appred{\dpred}{\su(\ii)} \rightarrow \falseform$.
However, $p(\paramA) \wedge \neg p(\paramB) \wedge \appred{\dpred}{\paramB}$ is satisfiable (with $\paramA^I \not = 0$).
In the context in which the rule is applied however, it will be ensured that satisfiability is preserved.
The intuition is that if an equation such as $\paramA \iseq 0$ occurs in the node, then $\paramA$ must have been instantiated previously, hence the term represented by $\paramA$ must be of a depth strictly greater than $\n$.
Due to the chosen instantiation strategy, all parameters of depth greater or equal to that of $\paramA$,
must have been instantiated (this property is not fulfilled by the previous formula: $\paramB$ should be instantiated since its depth is at most $1$ by definition).
Then it may be seen that the interpretation of the remaining formul{\ae} does not depend on the value of  $\paramA$, since the depth of their indices must be strictly less than that of $\paramA$.
Note that the removal of equations is \emph{essential} for ensuring termination. }

\label{sect:proof}

In this section, we present our procedure for testing the
$\classint$-satisfiability of admissible formul{\ae}.
We employ a tableaux-based procedure, with several kinds of inference rules:
\emph{Decomposition rules}
that reduce each formula to a conjunction of \bforms, equational literals, and defined literals;
\emph{Unfolding rules} that allow to unfold the defined atoms (by applying the rules in $\rules$);
\emph{Equality rules} for reasoning on equational atoms;
and
\emph{Delayed instantiation schemes} that replace a
parameter $\paramA$ by some term $f(\paramB_1,\dots,\paramB_n)$, where $f$ is a constructor and $\paramB_1,\dots,\paramB_n$ are new constant symbols.
We consider proof trees labeled by sets of formul{\ae}.
If $\alpha$ is a node in a tree $\T$ then $\labelsof{\alpha}{\T}$ denotes the  label of $\alpha$.
A node is \emph{closed} if it contains $\falseform$.
As usual, our procedure is specified by a set of \emph{expansion rules} of the form
\begin{tabular}{c}
$\Psi$ \\
\hline
$\Psi_1$ \quad \vline \quad \dots \quad \vline \quad $\Psi_n$
\end{tabular}
with $n \geq 1$,
meaning that a non-closed leaf node labeled by a set $\Phi \supseteq \Psi$ (up to a substitution of the meta-variables)
may be expanded by adding $n$ children labeled by $(\Phi \setminus \Psi) \cup \Psi_1$, \dots, $(\Phi \setminus \Psi) \cup \Psi_n$ respectively.
 We assume moreover that the formul{\ae} $\Psi_1,\dots,\Psi_n$ have not already been  generated in the considered branch (to avoid  redundant applications of the rules).
 For any tree $\T$, we write $\alpha \parent{\T} \beta$ iff $\beta$ is a child of $\alpha$.
$\parent{\T}^*$ denotes as usual the reflexive and transitive closure of $\parent{\T}$.

We need to introduce some additional notations and definitions.
For any interpretation $I$ and for any element $v$ in the domain of $I$, we denote by $\mint{v}{I}$ the depth of the constructor term denoted by $v$, formally defined as follows:
$\mint{v}{I} = 0$ if $v$ is in $D_{\sort}$ and $\sort \not \in \indsort$, otherwise
$\mint{\valof{f(t_1,\dots,t_n)}{I}}{I} = 1+\max(\{ \mint{\valof{t_i}{I}}{I} \mid i \in [1,n] \})$,
     with the convention that $\max(\emptyset) = 0$.
It is easy to check that the function $v \mapsto \mint{v}{I}$ is well-defined,
 for every interpretation $I \in \classint$.

For the sake of readability, we shall assume that there exists a function symbol $\depth$ such that:
$\intof{\depth}{I}(v) \isdef \mint{v}{I}$.
 The formula $\maxs(E) \iseq t$ (where $E$ is a finite set of terms) is written as a shorthand for
 $\bigwedge_{s \in E} (s \leq t) \wedge \bigvee_{s \in E} (s \iseq t)$ if $E \not = \emptyset$ and for $0 \iseq t$ if $E = \emptyset$.

Let $\T$ be a tree and let $\alpha$ be a node in $\T$.
A parameter $\paramA$ is {\em \solved} in $\alpha$ if the only formula of $\T(\alpha)$ containing $\paramA$ is of the form $\paramA \iseq \paramB$ where $\paramB$ is a parameter. An equation $\paramA\iseq \paramB$ is {\em \solved} in $\alpha$ if $\paramA$ is \solved. Notice that $\iseq$ is \emph{not} considered as commutative.
For every set of formul{\ae} $\Phi$, $\eq{\Phi}$ denotes the set of equations in $\Phi$ and
$\noneq{\Phi} \isdef \Phi \setminus \eq{\Phi}$.
A {\em \renaming} is a function $\ren$ mapping every parameter to a parameter of the same sort, such that $\ren(\n) = \n$.
Any \renaming $\ren$ can be extended into a function mapping
every formula $\phi$ to a formula $\ren(\phi)$, obtained by replacing
every parameter $\paramA$ occurring in $\phi$ by $\ren(\paramA)$.
Let $\Phi$ and $\Psi$ be two sets of formul{\ae}.
We write $\Phi \moregen \Psi$ iff
there exists a \renaming $\ren$
such that $\ren(\Psi) \subseteq \Phi$.

A {\em \prooftree} for $\phi$ is a tree constructed by the  rules of Figure \ref{fig:rulesA} below and such that the root is obtained by applying  \startrule on $\phi$.
We assume that  \veerule and \wedgerule
are applied with the highest priority.

\newcommand{\examplesrules}{
We provide a simple example to illustrate the rule applications. {\small
\begin{example}
Consider
the formula $\forall x\, \neg p(x) \wedge \appred{\dpred}{\paramA}$, together with the rules: $\appred{\dpred}{a} \rightarrow p(b)$
and $\appred{\dpred}{f(x,y)} \rightarrow \appred{\dpred}{x} \wedge \appred{\dpred}{y}$
(where $\C = \{ a\antispace\antispace:\antispace\antispace\sort,\, f\antispace\antispace:\antispace\antispace\sort\antispace \times\antispace \sort \antispace\rightarrow \antispace\sort,\, 0,\, \su \}$ and $\profile(\paramA) = \sort$). The root formula is $\forall x\, \neg p(x)  \wedge \appred{\dpred}{\paramA} \wedge \max(\{ \depth(\paramA) \}) \iseq \n$.
By normalization using \wedgerule
 we get
$\{ \forall x\, \neg p(x), \appred{\dpred}{\paramA}, \depth(\paramA) \iseq \n \}$. No rule applies, except \depthexplrule, which replaces $\n$ by $\su(0)$ or $\su(\n)$. In both cases, \explrule applies on $\paramA$.
 In the first branch, the rule adds the formula $\paramA \iseq a$
 and in the second one, it yields $\paramA \iseq f(\paramB,\paramC)$ (where $\paramB,\paramC$ are fresh parameters).
 In the former branch, \inductionrule replaces the formula $\appred{\dpred}{\paramA}$ by $p(b)$, then an irreducible node is reached.
 In the latter branch, the formul{\ae} $\appred{\dpred}{\paramB}$ and $\appred{\dpred}{\paramC}$ are inferred.
 Then \looprule applies, using the renaming: $\ren(\paramA)= \paramB$ or $\ren(\paramA) = \paramC$, hence the node is closed. The only remaining (irreducible) node is $\{ p(b), \forall x\, \neg p(x) \}$. The unsatisfiability of this set of formul{\ae} can be easily checked.
\end{example}
}

 The following example shows evidence of the importance of the depth rules:

{\small
 \begin{example}
 Consider the formula: $p(\paramA) \wedge \appred{\dpred}{\paramA} \wedge \appred{\cpred}{\paramB}$ with the rules $\appred{\dpred}{\su(x)} \rightarrow \appred{\dpred}{x}, \appred{\dpred}{0} \rightarrow \trueform$, $\appred{\cpred}{\su(x)} \rightarrow \falseform$ and $\appred{\cpred}{0} \rightarrow \neg p(0)$.
If the parameters were instantiated in an arbitrary order, then one could choose for instance to instantiate $\paramA$ by $\su(\paramA')$, yielding an obvious loop (indeed, the unfolding of $\appred{\dpred}{\paramA}$ yields $\appred{\dpred}{\paramA'}$, thus it suffices to consider the renaming $\ren(\paramA) = \paramA'$ and $\ren(\paramB)= \paramB$). Then the only remaining branch corresponds to the case
$\paramA \iseq 0$, which is actually unsatisfiable.
This trivial but instructive example shows that reasoning on the depth of the parameters is necessary to ensure that the model will eventually be reached. In this example, the depth of $\paramA$ is maximal and that of $\paramB$ is not, e.g.: $\paramA \iseq \su(0)$ and $\paramB \iseq 0$.
The problem stems from the fact that \looprule is \emph{not} sound in general, since equational atoms are removed from the formul{\ae} before testing for subsumption (the removal of such atoms is \emph{crucial} for  termination).
\end{example}
}}
Most of the rules in in Figure \ref{fig:rulesA} are self-explanatory.
We only briefly comment on some important points.

\commentstart

\commentdecomp

\commentinduction

\commenteq

\commentdepth

\commentexpl

\commentloop

\examplesrules

{\small

\begin{figure}
\fbox{
\begin{minipage}{1\textwidth}
\depictstart
\rule{\linewidth}{0.3pt}
\depictdec
\rule{\linewidth}{0.3pt}
\begin{tabular}{ccc}
\depictclos \quad \quad & \quad \quad
\depicteclash & \quad \quad
\depictnclos
\end{tabular}

\rule{\linewidth}{0.3pt}
\depictinduction

\rule{\linewidth}{0.3pt}
\depicteqdiseq
\rule{\linewidth}{0.3pt}
\depictrep
\rule{\linewidth}{0.3pt}
\begin{tabular}{ll}
\depictstrict \quad &
\quad \depictless
\end{tabular}
\rule{\linewidth}{0.3pt}
\begin{tabular}{ll}
\depictlessclash &
\depictsep
\end{tabular}
\rule{\linewidth}{0.3pt}
\depictexpl
\rule{\linewidth}{0.3pt}
\depictdepthexpl
\rule{\linewidth}{0.3pt}
\depictloop
\end{minipage}
}
\caption{Expansion rules \label{fig:rulesA}}
\end{figure}
}

\section{Properties of the Proof Procedure}

\label{sect:prop}

\newcommand{\soundmain}{This follows immediately from the previous lemmata.}

\newcommand{\soundaux}{
We begin by showing that \startrule preserves satisfiability:

\begin{lemma}
For every \prooftree $\T$ of root $\alpha$ for $\phi$, $\phi$ is $\classint$-satisfiable iff $\labelsof{\alpha}{\T}$ has a model $I \in \classint$.
\end{lemma}

\begin{proof}
By definition, $\labelsof{\alpha}{\T}$ is of the form
$\{ \phi \} \cup \{ \max(\{ \depth(\paramA_i) \mid i \in [1,n] \})  \lesseq  \n \}$, where
$\{ \paramA_1,\dots,\paramA_n \}$ is the set of parameters occurring in $\phi$
and $\n$ does not occur in $\phi$.
Obviously, if $\labelsof{\alpha}{\T}$ is satisfiable, then $\phi$ also is. Conversely, let $I$ be a model of $\phi$.
Let $J$ be the interpretation coinciding with $I$, except for the interpretation of $\n$ that is defined as follows:
$$\valof{\n}{J} \isdef \max \{ \mint{\valof{\paramA_i}{I}}{I} \mid i \in [1,n] \}$$
Since $I$ and $J$ coincide on every symbol occurring in $\phi$ we must have $J \models \phi$. Furthermore, since $I$ and $J$ have the same domain and coincide on every constructor symbol, we must have $\mint{v}{I} = \mint{v}{J}$ for every element $v$. Consequently,
for every $i \in [1,n]$ we have: $\mint{\valof{\paramA_i}{J}}{J} = \mint{\valof{\paramA_i}{J}}{I} = \mint{\valof{\paramA_i}{I}}{I}$ (since $\paramA_i \not = \n$), hence
$J \models \depth(\paramA_i) \lesseq \n$.
Thus $J \models \labelsof{\alpha}{T}$.
\end{proof}

We then show that most expansion rules preserve logical equivalence.

\begin{lemma}
\label{lem:invert}
The rules: \invertiblerules
are sound and invertible, i.e. for every \prooftree
$\T$ and for every node $\alpha$ in $\T$
on which one of these rules is applied, we have, for every interpretation $I \in \classint$:
\[\text{$I \models \labelsof{\alpha}{\T}$ iff
$\exists \beta, \beta \child{\T} \alpha \wedge I \models \labelsof{\beta}{\T}.$} \]
\end{lemma}

\begin{proof}
We consider each rule separately.
\bitem
\item{\textbf{\decomprules.} The proof is straightforward.}

\item{\textbf{\equalityrules.}
\bitem
\item{\Eqrule: The node $\alpha$ is labeled by
$\Phi \cup \{ \paramA \iseq f(\paramA_1,\dots,\paramA_n), \paramA \iseq g(\paramB_1,\dots,\paramB_m) \}$ and has only one child $\beta$ labeled by
$\Phi \cup \{ \decompe{f(\paramA_1,\dots,\paramA_n)}{g(\paramB_1,\dots,\paramB_m)}, \paramA \iseq f(\paramA_1,\dots,\paramA_n)\}$.
Obviously, $\labelsof{\alpha}{\T} \equiv
\Phi \cup \{ f(\paramA_1,\dots,\paramA_n) \iseq g(\paramB_1,\dots,\paramB_n), \paramA \iseq f(\paramA_1,\dots,\paramA_n)\}$.
By Condition \ref{int:decomp} in Definition \ref{def:int}, we have
$f(\paramA_1,\dots,\paramA_n) \iseq g(\paramB_1,\dots,\paramB_m) \equiv_{I} \decompe{f(\paramA_1,\dots,\paramA_n)}{g(\paramB_1,\dots,\paramB_m)}$.
Thus $\labelsof{\alpha}{\T} \equiv_{I} \labelsof{\beta}{\T}$.}
\item{\Diseqrule: The proof is similar.}
\item{\Seprule: We have $\paramA \iseq \paramB \vee \paramA \not \iseq \paramB \equiv \trueform$, hence the proof is immediate.}
\item{\Reprule: Obviously, $\phi \wedge \paramA \iseq \paramB \equiv \replaceby{\Phi}{\paramA}{\paramB} \wedge \paramA \iseq \paramB$.}

\item{\Eclashrule: By definition, $\paramA \not \iseq \paramA \equiv \falseform$.}
\eitem}

\item{\textbf{\depthrules.}
\bitem
\item{\Strictrule:
By definition of the interpretation of $\lesseq$ and $\less$, we have
    $\depth(\paramA) \lesseq \n \equiv (\depth(\paramA) \iseq \n \vee \depth(\paramA) \less \n)$.
     }

\item{\Lessclashrule:
By definition of the interpretation of $\less$, if $I \models \paramA \iseq \paramB$ then $I \models \depth(\paramA)\iseq \depth(\paramB)$, thus $\depth(\paramA) \less \n \wedge \depth(\paramB) \iseq \n \equiv \depth(\paramA) \less \n \wedge \depth(\paramB) \iseq \n \wedge \paramA \not \iseq \paramB$.}
\item{\Lessrule: By definition of the interpretation of $\less$ and $\su$, we have $\depth(\paramA) \less \su(\n) \equiv \depth(\paramA) \lesseq \n$.}

\eitem}
\item{\textbf{\inductionrules.}
\bitem
\item{\Inductionrule:
The node $\alpha$ is labeled by a set of formul{\ae} $\Phi \cup
\{\appred{\dpred}{\paramA} \} \cup \{
\paramA\iseq f(\vec{\paramB})  \}$.
Moreover, $\alpha$ has only one child $\beta$ labeled by: $\Phi \cup \{ \psi \} \cup \{ \paramA\iseq f(\vec{\paramB})
 \}$, where $\psi$ is the formula obtained
from $\predeft{\dpred}{f(\vec{\paramB})}$
by replacing every occurrence of $f(\vec{\paramB})$ by $\paramA$.
If $I \not \models \paramA \iseq f(\vec{\paramB})$ then
we have obviously $\labelsof{\alpha}{T} \equiv_{I} \labelsof{\beta}{T} \equiv_{I} \falseform$.
Otherwise,
$\psi \equiv_I \predeft{\dpred}{f(\vec{\paramB})}$
and
$\appred{\dpred}{\paramA} \equiv_I  \appred{\dpred}{f(\vec{\paramB})}$.
Furthermore, by Condition \ref{int:def} in Definition \ref{def:int}, we have
$\appred{\dpred}{f(\vec{\paramB})}
\equiv_{I}
\predeft{\dpred}{f(\vec{\paramB})}$.
Thus $\labelsof{\alpha}{\T} \equiv_{I} \labelsof{\beta}{\T}$.
}
\eitem
}
\eitem

\end{proof}

We now prove that the remaining rules (except \looprule) preserve $\classint$-satisfiability.
We first need to analyze the form of the
formul{\ae} containing $\depth$ occurring in the
\prooftree:

\begin{lemma}
\label{lem:depthatom}
A \emph{$\depth$-atom} is an atom containing the $\depth$ function symbol.
Let $\T$ be a \prooftree and let $\alpha$ be a node in $\T$.
If $\phi$ is $\depth$-atom occurring in a formula $\psi \in \labelsof{\alpha}{\T}$ then:
 \begin{itemize}
 \item{$\psi$ is a boolean combination of $\depth$-atoms.}
  \item{$\phi$ is of the form
  $\depth(\paramA) \ordrel t$, where $\ordrel \in \{ \iseq, \less, \lesseq \}$ and
$t \in \{ \n, \su(\n), \su(0) \}$.}
\item{If $\alpha$ is a \layer, then
$\ordrel \in \{ \iseq, \less \}$.}
\end{itemize}
\end{lemma}

\begin{proof}
The only rules that can introduce formul{\ae} containing $\depth$ are \startrule, \strictrule, \lessrule
and \explrule.
It is clear, by inspection of these rules, that the added formul{\ae} fulfill the above properties.
Moreover, if $\alpha$ is a \layer then by irreducibility w.r.t. \veerule and \wedgerule, $\psi$ must be an atom.
By irreducibility w.r.t. \strictrule, $\ordrel$ cannot be $\lesseq$ and
by irreducibility w.r.t. \explrule, $t$ must be $\n$.
\depthexplrule can affect the right-hand side of a $\depth$-atom by replacing $\n$ by $\su(\n)$ or $\su(0)$. However, due to the control, this rule is only applied on \layers, thus the right-hand side must be $\n$, hence no formula of the form $\su(\su(t))$ can be introduced.
\end{proof}

\begin{lemma}
\label{lem:sat}
The rules \satrules
preserve satisfiability: for every \prooftree
$\T$ and for every node $\alpha$ in $\T$
on which one of these rules is applied and
for every interpretation $I \in \classint$,
the following properties are equivalent:
\bitem
\item{
$I \models \labelsof{\alpha}{T}$.}
\item{There exists $\beta \child{\T} \alpha$ and $J \in \classint$ such that the following conditions hold:
     \bitem
     \item[$\bullet$]{$J \models \labelsof{\beta}{T}$.}
     \item[$\bullet$]{For every symbol $s$ distinct from $\n$ and occurring in $\T(\alpha)$, we have $\intof{s}{J}=\intof{s}{I}$.}
     \item[$\bullet$]{If \depthexplrule is applied on $\alpha$ then $\intof{\n}{J}=\intof{\n}{I}-1$, otherwise $\intof{\n}{J}=\intof{\n}{I}$.}
     \eitem
}
\eitem
\end{lemma}

\begin{proof}
Again, we need to distinguish two cases.
\bitem
\item{\Explrule:
By definition, $\alpha$ is labeled by $\Phi \cup \{ \depth(\paramB)\iseq \su(t) \}$ and its unique child $\beta$ is labeled by $\Phi \cup \{ \bigvee_{i \in [1,n]} \maxs(E_i)\iseq t \wedge \paramB \iseq t_i \}$, where $t_i$ is of the form $f_i(\vec{A}_i)$ and $E_i$ is the set of terms $\depth(\paramC)$ where $\paramC$ is a component of $\vec{A}_i$ of a sort in $\indsort$.
Let $I$ be a model of $\labelsof{\alpha}{\T}$.
By Point \ref{int:ind} in Definition \ref{def:nice}, $\valof{\paramB}{I}$ is equal to $\valof{f_i(\vec{s})}{I}$ for some $i \in [1,n]$ and for some vectors of terms $\vec{s}$.
Let $J$ be the interpretation coinciding with $I$, except on the constant symbols of $\vec{\paramA}_i$ that are interpreted in such a way that $\valof{\vec{\paramA_i}}{J} = \valof{\vec{s}}{I}$ (this is possible since $\vec{\paramA_i}$ is a vector of fresh, distinct, constant symbols).
We have $J \models \paramB \iseq t_i$. Furthermore, by definition,
$\mint{\paramB}{J}= 1+\max \{ \mint{\paramC}{J} \mid \mbox{$\paramC$ occurs in $\vec{\paramA_i}$} \}$.
Thus, since we have $\mint{v}{J}=0$ if $v$ is of a sort in $\S \setminus \indsort$,
$\mint{\paramB}{J}= 1+\max_{\paramC \in E_i} \mint{\paramC}{J}$.
But since $I \models \depth(\paramB)\iseq \su(t)$, we have $\mint{\paramB}{I} = \valof{t}{I}+1$, thus
$\mint{\paramB}{J} = \valof{t}{I}+1$
and
$J \models \maxs(E_i) \iseq t$.
Hence $J \models \labelsof{\beta}{\T}$.

Conversely, if $J \models \labelsof{\beta}{\T}$, then $J \models \paramB \iseq t_i$, for some $i \in [1,n]$. Then by definition of $\mint{t_i}{I}$, we have $\mint{t_i}{J}=1+\maxs(E_i)$, hence since $J \models \max(E_i)\iseq t$, we have $J \models \depth(\paramB)\iseq \su(t)$.
}

\item{\Depthexplrule:
$\alpha$ is labeled by $\Phi$ and has two children, $\beta_1$ and $\beta_2$, labeled respectively by $\replaceby{\Phi}{\n}{\su(0)}$ and
$\replaceby{\Phi}{\n}{\su(\n)}$.
Let $I$ be an interpretation validating $\Phi$.
By definition, $\n$ occurs in $\Phi$, which means that
$\Phi$ contains a formula of the form $\depth(\paramA) \ordrel t$, where $\n$ occurs in $t$.
By Lemma \ref{lem:depthatom}, $t$ must be $\n$, thus, since $I \models \depth(\paramA) \ordrel \n$, necessarily $\valof{\n}{I} > 0$ (since $\mint{\paramA}{I}>0$).
 If $\valof{\n}{I}=\su(0)$ then obviously $I \models \labelsof{\beta_1}{\T}$. Otherwise, let $J$ be an interpretation coinciding with $I$ except that
$\valof{\n}{J} \isdef \valof{\n}{I}-1$. Obviously, we have
$J \models \replaceby{\Phi}{\n}{\su(\n)}$, thus $J \models \labelsof{\beta_2}{\T}$.
The converse is immediate.
}

\eitem

\end{proof}

We write $\alpha \parentn{\T}{k} \beta$ iff $\alpha \parent{\T}^* \beta$ and there exists exactly $k$ applications of \depthexplrule in the branch from $\alpha$ to $\beta$.

\begin{corollary}
\label{cor:satn}
Let $\T$ be a \prooftree. Then:
\bitem
\item{If $I \models \labelsof{\beta}{\T}$ and $\alpha \parentn{\T}{k} \beta$  then $I[(\valof{\n}{I}+k)/\n] \models \alpha$.}
\item{If $I \models \labelsof{\alpha}{\T}$ then there exists a leaf $\beta$ such that $\alpha \parentn{\T}{k} \beta$ and an interpretation $J$ such that $J \models \labelsof{\beta}{\T}$, $I$ and $J$ coincide on any symbol occurring in $\labelsof{\alpha}{\T}$ distinct from $\n$
    and $\valof{\n}{J} = \valof{\n}{I}-k$.}
\eitem
\end{corollary}

\begin{proof}
This is an immediate consequence of Lemmata \ref{lem:invert} and \ref{lem:sat}.
\end{proof}

 There only remains to handle the case of the \looprule rule, which is actually the most complex one.
To this aim, we need to introduce some additional definitions and lemmata.

\newcommand{\controled}{$\n$-controlled\xspace}

\begin{definition}
A parameter $\paramA$ is {\em \instantiated} in a node $\alpha$ of a \prooftree $\T$
iff   $\labelsof{\alpha}{\T}$ contains a formula of the form $\paramA \iseq f(\vec{\paramB})$.
It is {\em \controled} if $\labelsof{\alpha}{\T}$ contains a formula of the form $\depth(\paramA) \ordrel t$ with $\ordrel \in \{ \less, \iseq, \lesseq \}$.
 \end{definition}

\newcommand{\decomposed}{decomposed\xspace}

\begin{definition}
A node that is irreducible by \veerule and \wedgerule
is \emph{decomposed}.
\end{definition}

We write $\alpha \parentd{\T} \beta$ if
$\alpha$ is non-\decomposed and $\alpha \parent{\T} \beta$. Due to the control, $\beta$ is obtained by applying \veerule or \wedgerule.

\begin{proposition}
\label{prop:parent}
Let $\T$ be a \prooftree.
Let $\alpha \parent{\T} \beta$.
\begin{enumerate}
\item{
If $\paramA$ is \instantiated in $\alpha$ and not \solved in $\beta$ then it is also \instantiated in $\beta$.\label{parent:free}}
\item{
If $\paramA$ is \controled in $\alpha$ and if $\alpha \parentd{\T} \beta$ then $\paramA$ is \controled in $\beta$.\label{parent:controled}}
\end{enumerate}
\end{proposition}

\begin{proof}
\begin{enumerate}
\item{
If $\paramA$ is \instantiated in $\alpha$ then  $\labelsof{\alpha}{\T}$ contains a formula of the form $\paramA \iseq f(\vec{B})$.
Since $\paramA$ cannot be replaced, it is easy to check (by inspection of the expansion rules) that no rule can remove such a formula (except \Eqrule, but in this case another formula of the form $\paramA \iseq g(\vec{\paramC})$ occurs in the node).
 Thus $\paramA$ is \instantiated in $\beta$.
}
\item{This is immediate since \veerule and \wedgerule
    cannot delete non-complex formul{\ae}.}
\end{enumerate}

\end{proof}

\begin{lemma}
\label{lem:depth}
Let $\T$ be a \prooftree. Let $\alpha$ be a non-closed \decomposed node in $\T$.
Every parameter distinct from $\n$ occurring in $\labelsof{\alpha}{\T}$ that is neither \solved nor \instantiated is \controled in $\alpha$.
\end{lemma}

\begin{proof}
The proof is by induction on
the depth of $\alpha$ in $\T$.

Assume first that all the parent nodes of $\alpha$ are non-\decomposed.
Since \veerule and \wedgerule
are applied with the highest priority, this implies that $\gamma \parentd{\T}^* \alpha$, where $\gamma$ is the root of $\T$.
These rules cannot introduce new parameters hence $\paramA$ occurs in $\labelsof{\gamma}{\T}$.
Thus, by definition of \startrule, $\labelsof{\gamma}{\T}$ must contain exactly one formula of the form $\depth(\paramA) \lesseq \n$. Hence $\paramA$ is \controled in $\gamma$. By Proposition \ref{prop:parent} (Point \ref{parent:controled}), it must be \controled in $\alpha$.

Now assume that there exists a node $\beta \parent{\T}^{+} \alpha$ that is irreducible by \veerule and \wedgerule.
We assume that $\beta$ is the deepest node having this property. Then there exists a node $\lambda$ such that $\beta \parent{\T} \lambda \parentd{\T}^* \alpha$.
We distinguish two cases.
\bitem
\item{Assume that $\paramA$ occurs in $\labelsof{\beta}{\T}$.
By Proposition \ref{prop:parent} (Point \ref{parent:free}), $\paramA$ is also \free in $\beta$. By the induction hypothesis, $\beta$ contains a formula $\depth(\paramA) \ordrel t$.
If $\paramA$ is \controled in $\lambda$ then the proof follows immediately from Proposition \ref{prop:parent} (Point \ref{parent:controled}). Now assume that $\paramA$ is not \controled in $\lambda$, i.e., that the rule applied to $\beta$ deletes
the formula $\depth(\paramA) \ordrel t$.
By inspection of the expansion rules, it can be seen that the only rules that can delete such a formula
are \lessrule, \strictrule, and \explrule (note that \depthexplrule can only affect $t$  and since $A$ occurs in $\alpha$, \reprule cannot be applied on $\beta$).
\lessrule replaces a formula $\depth(\paramA) \less \su(\n)$ by $\depth(\paramA) \lesseq \n$, hence $\paramA$ is \controled in $\lambda$, which is impossible by assumption.
If \strictrule is applied on $\depth(\paramA) \ordrel t$ then a formula of the form $\depth(\paramA) \iseq \n \vee \depth(\paramA) \lesseq \n$ occurs in $\labelsof{\lambda}{\T}$. Since $\alpha$ is \decomposed, $\labelsof{\alpha}{\T}$ contains either
$\depth(\paramA) \iseq \n$ or $\depth(\paramA) \lesseq \n$, hence $\paramA$ is \controled in $\alpha$.
If \explrule is applied and deletes $\depth(\paramA) \ordrel t$ it must simultaneously introduce an equation of the form $\paramA \iseq f(\dots)$, thus $\paramA$ is \instantiated in $\lambda$, hence, by Proposition \ref{prop:parent} (Point \ref{parent:free}), also in $\alpha$, a contradiction.}
\item{Now, assume that $\paramA$ does not occur in $\labelsof{\beta}{\T}$.
The only rule that can introduce a new parameter $\paramA$ is \explrule, but this rule simultaneously introduces a formula $\maxs(E)\iseq \n$, where $\paramA \in E$. After some decomposition steps, an atom of the form $\depth(\paramA) \iseq \n$ or
$\depth(\paramA) \lesseq \n$ must occur in every branch.
Thus the property remains true.}
 \eitem
 \end{proof}

\newcommand{\imapI}[3]{\imap(##1,##3,##2)}
\newcommand{\tI}[3]{J(##1,##3,##2)}

\begin{definition}
Let $I$ be an interpretation, let $\paramA$ be a parameter of sort $\sort$
and let $v$ be an element of $\intof{\sort}{I}$.
 We denote by $\imapI{I}{v}{\paramA}$ the \imapping\ for $I$ such that $\imap(v)\isdef\intof{\paramA}{I}$ and $\imap(e) \isdef e$ for every $e \not = v$, and by $\tI{I}{v}{\paramA}$ the interpretation obtained from
$\imap(I)$ by replacing the value of every parameter $\paramB$ such that $I \models \paramB\iseq \paramA$ by $v$.
\end{definition}

\begin{proposition}
\label{prop:bforms}
Let $\phi$ be a \bform.
Let $I$ be an interpretation, let $\paramA$ be a parameter of sort $\sort$
and let $v$ be an element of $\intof{\sort}{I}$.
If for all parameters $\paramB$ occurring in $\phi$, we have $\intof{\paramB}{I} \not = v$, then:
\begin{enumerate}
\item{For every term $t$ of a sort $\sort' \not \in \indsort$ occurring in $\phi$, we have $\valof{t}{\tI{I}{v}{\paramA}} = \valof{t}{I}$.
}
\item{For every subformula $\psi$ of $\phi$, $\valof{\psi}{\tI{I}{v}{\paramA}} = \valof{\psi}{I}$.}
\end{enumerate}

\end{proposition}

\begin{proof}
The proof is by structural induction on $t$ and $\psi$. We only give the detailed proof for $t$, since the inductive cases for $\psi$ are straightforward (since \bforms cannot contain equations between terms of a sort in $\indsort$).

Let $J = \tI{I}{v}{\paramA}$ and $\imap = \imapI{I}{v}{\paramA}$.
If $t$ is a variable, then $I$ and $J$ coincide on $x$, thus we have $\valof{t}{J}=\valof{t}{I}$.
Assume that $t$ is of the form $f(t_1,\dots,t_n)$, where $f$ is a function symbol of profile $\sort_1 \times \dots \times \sort_n \rightarrow \sort$.
By definition of $\imap(I)$, we have $\valof{t}{J} = \intof{f}{I}(\imap(\valof{t_1}{J}),\dots,
\imap(\valof{t_n}{J}))$.

By the induction hypothesis,
for every $i \in [1,m]$, if $\sort_i \not \in \indsort$ then
$\valof{t_i}{J} = \valof{t_i}{I}$, thus
$\imap(\valof{t_i}{J})= \valof{t_i}{I}$ (since $\imap$ is the identity on any element distinct from $v$, hence on any element of the domain of a sort non occurring in $\indsort$).

Now, assume that there exists $i \in [1,n]$ such that $t_i$ is of \aninductivesort. Since $t$ occurs in $\phi$, it cannot contain any constructor symbol (by Condition \ref{niceform:noconst} in Definition \ref{def:niceform}), hence $t_i$ must be a parameter.
If $I \models t_i\iseq \paramA$, then
$\imap(\valof{t_i}{J}) =
\imap(v) = \valof{A}{I} = \valof{t_i}{I}$.
Otherwise
$\imap(\valof{t_i}{J})
= \imap(\valof{t_i}{I}) = \valof{t_i}{I}$.

Thus for all $i \in [1,n]$, $\imap(\valof{t_i}{J}) =
\valof{t_i}{I}$ and $\valof{t}{J} = \intof{f}{I}(\valof{t_1}{I},\dots,
\valof{t_n}{I}) = \valof{t}{I}$.

\end{proof}

\begin{proposition}
\label{prop:depth}
Let $\T$ be a \prooftree. Let $\alpha$ be a layer in $\T$.
Let $I$ be a model of $\labelsof{\alpha}{\T}$.
If $\paramA$ is neither \solved nor \instantiated in $\alpha$ then $\mint{\valof{\paramA}{I}}{I}\leq \valof{\n}{I}$.
If $\paramA$ is \instantiated in $\alpha$ then $\mint{\valof{\paramA}{I}}{I}>\valof{\n}{I}$.
\end{proposition}

\begin{proof}
The first point is a direct consequence of Lemma \ref{lem:depth}.
Let $\paramA$ be a parameter that is \instantiated in $\alpha$.
Then $\labelsof{\alpha}{\T}$ must contain a formula of the form $\paramA \iseq f(\paramB_1,\dots,\paramB_n)$.
The only rule that can introduce such a formula is \explrule.
Thus there must exist a node $\beta \parent{\T}^* \alpha$
on which \explrule is applied, yielding a formula of the form
$\paramA' \iseq f(\paramB_1',\dots,\paramB_n')$.
Furthermore, $\paramA'$ must be reduced to $\paramA$ by \reprule, hence there exist $k$ nodes $\paramA_1,\dots,\paramA_k$ with $\paramA_1=\paramA, \paramA_k=\paramA'$ and for all $i \in [1,k-1]$, there exists a node $\gamma_i$ such that $\beta \parent{\T}^+ \gamma_i \parent{\T}^+ \alpha$ and $\paramA_i\iseq \paramA_{i+1} \in \labelsof{\gamma_i}{\T}$.
By definition of \explrule, $\labelsof{\beta}{\T}$
contains a formula $\depth(\paramA')\iseq \su(\n)$.
By Corollary \ref{cor:satn}, there exists $l \in \N$ such that $I[\valof{\n}{I}+l/\n]$ validates the formula $\depth(\paramA') \iseq \su(\n)$ and all the formul{\ae} $\paramA_i \iseq \paramA_{i+1}$ ($1 \leq i \leq k-1$).
 Then we must have $\mint{\paramA}{I} = \valof{\su(\n)}{I}+l > \valof{\n}{I}$.
\end{proof}

\begin{proposition}
\label{prop:layersolved}
Let $\T$ be a \prooftree.
Let $\alpha$ be a \layer in $\T$.
Any equation between parameters occurring in $\labelsof{\alpha}{\T}$ is \solved.
\end{proposition}

\begin{proof}
If $\labelsof{\alpha}{\T}$ contains a non-\solved equation $\paramA \iseq \paramB$ then by definition \reprule would apply.
\end{proof}

\begin{lemma}
\label{lem:pur}
Let $\T$ be a \prooftree.
Let $\alpha$ be a \layer in $\T$.
If $I \models \noneq{\labelsof{\alpha}{\T}}$ then there exists an interpretation
$J$ such that $J \models \labelsof{\alpha}{\T}$ and $\valof{\n}{J}=\valof{\n}{I}$.
\end{lemma}

\begin{proof}
We denote by $\purify{\Phi}$ the set obtained from $\Phi$ by removing all formul{\ae} of the form $\paramA \iseq f(\vec{\paramB})$.

If $I \models \noneq{\labelsof{\alpha}{\T}}$ then it is obvious that there exists an interpretation $I'$ such that
$I' \models \purify{\labelsof{\alpha}{\T}}$: indeed, all the formul{\ae} occurring in $\purify{\labelsof{\alpha}{\T}}$, but not in $\noneq{\labelsof{\alpha}{\T}}$
are equations between parameters, which must be \solved by Proposition \ref{prop:layersolved}. Thus it suffices to interpret each solved parameter $\paramA$ in the same way as the -- necessarily unique -- parameter $\paramB$ such that $\paramA \iseq \paramB$ occurs in $\purify{\labelsof{\alpha}{\T}}$.

$\n$ cannot be \solved, thus $\valof{\n}{I'}=\valof{\n}{I}$.
By definition, the \solved parameters cannot occur in $\noneq{\labelsof{\alpha}{\T}}$, thus
$I$ and $I'$ coincide on any formula in $\noneq{\labelsof{\alpha}{\T}}$
and $I' \models \noneq{\labelsof{\alpha}{\T}}$.
Moreover, $I'$ validates all \solved equations, by definition.

Let $>$ be a total order on parameters such that
$\paramA > \paramB$ if $\mint{\paramA}{I} > \mint{\paramB}{I}$.

For any parameter $\paramA$, we denote by $\purifybis{\Phi}{\paramA}$ the set of formul{\ae} obtained from $\Phi$ by deleting all formul{\ae} of the form $\paramB \iseq f(\vec{\paramB})$
where $\paramB > \paramA$.
We shall show, by induction on $\paramA$, that one can construct an interpretation
$J$ such that $J \models \purifybis{\labelsof{\alpha}{\T}}{\paramA}$ and $\valof{\n}{J}=\valof{\n}{I'}$. Then the result will follow, simply by instantiating $\paramA$ with the $<$-maximal parameter.

Assume that $J$ has been constructed for the greatest parameter $\paramC$ such that $\paramA > \paramC$ (if $\paramA$ is minimal, then we simply take $J=I'$).
If $\labelsof{\alpha}{\T}$ contains no formula of the form $\paramA \iseq f(\paramB_1,\dots,\paramB_n)$ then obviously
$\purifybis{\labelsof{\alpha}{\T}}{\paramA} = \purifybis{\labelsof{\alpha}{\T}}{\paramC}$ and
$J \models \purifybis{\labelsof{\alpha}{\T}}{\paramA}$.
Thus we assume that $\labelsof{\alpha}{\T}$ contains such a formula. Since $\alpha$ is a \layer, by irreducibility w.r.t. \eqrule, this formula must be unique.
    Let $v = \valof{f(\paramB_1,\dots,\paramB_n)}{I}$.
    Let $\imap = \imapI{I}{\paramA}{v}$ and $K = \tI{J}{\paramA}{v}$.

    We first show that for all parameters
    in $\purifybis{\labelsof{\alpha}{\T}}{\paramA}$, we have $\valof{\paramA'}{J} \not = v$. Notice that, by definition, $\paramA'$ cannot be \solved in $\alpha$.
    If $\paramA'$ is \free then by Proposition \ref{prop:depth}, we have $\mint{\paramA'}{J}\leq \n$.
    Moreover, since $\paramA$ is \instantiated, we have, still by
    Proposition \ref{prop:depth}, $\mint{\paramA}{J}>n$, thus $\mint{v}{I} > n$ and $\valof{\paramA'}{J} \not = v$.
    If $\paramA'$ is \instantiated, $\labelsof{\alpha}{\T}$ contains a formula
    $\paramA' \iseq g(\paramB'_1,\dots,\paramB'_k)$.
    By irreducibility w.r.t. \seprule, $\labelsof{\alpha}{\T}$ contains $\paramA \not \iseq \paramA'$.
    By irreducibility w.r.t. \diseqrule,
    $\labelsof{\alpha}{\T}$ must contain a set of disequations $E$ between elements of $\paramB_1,\dots,\paramB_n,\paramB'_1,\dots,\paramB'_k$ such that $E \models f(\paramB_1,\dots,\paramB_n)\not\iseq g(\paramB'_1,\dots,\paramB'_k)$.
    But $E \subseteq \purifybis{\labelsof{\alpha}{\T}}{\paramA}$, thus $J \models E$, whence $\valof{\paramA'}{J} \not = \valof{\paramA}{J}$.

    By definition $K \models \paramA \iseq f(\paramB_1,\dots,\paramB_n)$. Let $\phi$ be a formula occurring in $\purifybis{\labelsof{\alpha}{\T}}{\paramA}$. We know that $J \models \phi$. We prove that $K \models \phi$.

    By Proposition \ref{prop:bforms}, if $\phi$ is a  \bform then $\valof{\phi}{J} = \valof{\phi}{K}$, thus $K \models \phi$.

    If $\phi$ is of the form $\depth(\paramA') \ordrel \n$, for some $\ordrel \in \{ \iseq, \lesseq, \less \}$ then by Proposition \ref{prop:depth}, $\paramA'$ cannot be \instantiated, hence $\paramA'\not =\paramA$ and $J,K$ coincide on $\phi$, thus $K \models \phi$.

    If $\phi$ is of the form $\paramA' = g(\paramB'_1,\dots,\paramB'_m)$ then we have $\paramA' < \paramA$, hence $\paramB'_1,\dots,\paramB'_m < \paramA$, thus $J$ and $K$ coincide on $\paramA,\paramB'_1,\dots,\paramB'_m$ and the proof is immediate.

If $\phi$ is of the form $\paramB \iseq \paramC$ or $\paramB \not \iseq \paramC$ then by definition of $K$ we have $K \models \phi$.
\end{proof}

\begin{proposition}
\label{prop:moregen}
If $\Phi \moregen \Psi$ and $I \models \Phi$ then there exists an interpretation $J$ such that $\valof{\n}{J} = \valof{\n}{I}$ and $J \models \Psi$.
\end{proposition}

\begin{proof}
By definition, we have $\ren(\Psi) \subseteq \Phi$, for some renaming $\ren$.
It suffices to consider the interpretation $J$ coinciding with $I$ except that every parameter $\paramA$ is mapped to $\valof{\ren(\paramA)}{I}$.
It is clear that for every expression $e$, $\valof{e}{J} = \valof{\ren(e)}{I}$. Since $I \models \Phi$ we have
$I \models \ren(\Psi)$, hence $J \models \Psi$.
\end{proof}

A \prooftree $\T$ is \emph{$\classint$-satisfiable} iff
there exists a leaf node $\alpha$ in $\T$ such that
$\labelsof{\alpha}{\T}$ is $\classint$-satisfiable.

\begin{lemma}
\label{lem:global}
\looprule preserves \emph{global satisfiability} i.e.
if $\T$ is $\classint$-satisfiable then any \prooftree $\T'$ obtained from $\T$ by applying \looprule is also $\classint$-satisfiable.
\end{lemma}

\begin{proof}
Let $\alpha$ be the node on which \looprule is applied.
$\T'$ is identical to $\T$ except that $\alpha$ has a child $\beta$ whose label contains $\falseform$.
Obviously, if $\T'$ is satisfiable then so is $\T$ (since all $\classint$-satisfiable leaves of $\T'$ are in $\T$).

Conversely, let $\gamma$ be the root of $\T$ and
let $I$ be a model of $\labelsof{\gamma}{\T}$ such that the interpretation of $\n$ is minimal (i.e. if $\valof{\n}{J} < \valof{\n}{I}$ then $J \not \models \labelsof{\gamma}{\T}$).
By Corollary \ref{cor:satn},
there exists a leaf $\alpha'$ in $\T$ and an interpretation $J$
such that $J \models \labelsof{\alpha'}{\T}$, $\gamma \parentn{\T}{k} \alpha'$ and $\valof{J}{\n} = \valof{I}{\n}-k$.
If $\alpha'$ is distinct from $\alpha$, then $\alpha'$ is a leaf in $\T'$ and the proof is immediate. Thus we assume that $\alpha'=\alpha$.
By definition of \looprule, there exists a node $\beta \parent{\T}^* \alpha$ such that
$\noneq{\Phi} \moregen \noneq{\Psi}$, with $\labelsof{\alpha}{\T} = \Phi$ and
$\labelsof{\beta}{\T}=\Psi$. By Proposition \ref{prop:moregen}, there exists an interpretation $J'$ such that $J' \models \purify{\Psi}$ and $\valof{\n}{J}=\valof{\n}{J'}$.
By Lemma \ref{lem:pur}, there exists an interpretation $J''$ such that $J'' \models \Psi$ and $\valof{\n}{J''}=\valof{\n}{J}$.

By definition, there exist $k'$ and $k'' > 0$ such that
$\gamma \parentn{\T}{k'} \beta \parentn{\T}{k''} \alpha$, where $k = k'+k''$.
By Proposition \ref{prop:moregen}, there exists an interpretation $K$ such that $K \models \labelsof{\alpha}{\T}$ and $\valof{\n}{K}=\valof{\n}{J} = \valof{\n}{I}-k$. By Corollary \ref{cor:satn}, $J''[(\valof{\n}{J}+k')/\n] \models \labelsof{\gamma}{\T}$. But the value of $\n$ in $J''[\valof{\n}{J}+k']$ is $\valof{\n}{I}-k+k' = \valof{\n}{I}-k''$. Since $k'' \not = 0$ this contradicts the minimality of $I$.
\end{proof}
}

\proofstate{\subsection{Soundness}}{}

\newcommand{\before}[1]{\triangleleft_{#1}}

This short section merely contains the theorems formalizing the main properties of the proof procedure. \aboutproofs{}
We first state
that the previous rules are sound.

\addtheowithlemmas{
\begin{theorem}
\label{theo:sound}
Let $\T$ be a \prooftree for a formula $\phi$.
If $\T$ is closed then $\phi$ is unsatisfiable.
\end{theorem}
}
{Theorem \ref{theo:sound}}
{\protect\soundaux}
{\protect\soundmain}

\proofstate{\subsection{Completeness}}{}

We then state that the procedure is complete, in the sense
that the satisfiability of every irreducible node can be tested by the procedure for \bforms.

\newcommand{\comp}{testable\xspace}

\newcommand{\prcomp}{
The first point follows from Lemma \ref{lem:pur}.

By Lemma \ref{lem:layerform}, we only have to prove that $\labelsof{\alpha}{\T}$ contains no $\depth$-atoms and no defined atoms.

Assume that $\labelsof{\alpha}{\T}$ contains an occurrence of $\n$. Then since no other rule applies, \depthexplrule must apply, which is impossible. Thus $\n$ does not occur in $\labelsof{\alpha}{\T}$.
This implies that $\labelsof{\alpha}{\T}$ contains no $\depth$-atoms. But then by Lemma \ref{lem:depth}, this implies that all non \solved parameters occurring in $\labelsof{\alpha}{\T}$ are \instantiated.

Assume that $\labelsof{\alpha}{\T}$ contains a defined symbol $\dpred$.
By irreducibility w.r.t. \veerule and \wedgerule, this defined symbol must occur in a formula $\dpred_{\paramA} \in \labelsof{\alpha}{\T}$. Since $\paramA$ is \instantiated then \inductionrule applies, which is impossible.
}

\newcommand{\complem}{
\begin{lemma}
\label{lem:layerform}
Let $\T$ be a \prooftree. Let $\alpha$ be a \layer in $\T$. Let $\phi$ be a formula in $\noneq{\labelsof{\alpha}{\T}}$. One of the following conditions holds:
\bitem
 \item{$\phi$ is a \bform.}
  \item{$\phi$ is of the form $\appred{\dpred}{\paramA}$ where $\paramA$ is a parameter.}
  \item{$\phi$ is of the form $\depth(\paramA) \ordrel \n$, where $\ordrel \in \{ \less, \iseq \}$.}
 \item{$\phi$ is of the form $\paramA \not \iseq \paramB$, where $\paramA,\paramB$ are parameters.}
  \eitem
Furthermore, if $\paramA$ and $\paramB$ are two non \solved parameters occurring in $\labelsof{\alpha}{\T}$
then $\paramA \not \iseq \paramB \in
\labelsof{\alpha}{\T}$.
\end{lemma}

\begin{proof}
Let $\phi$ be a formula occurring in $\noneq{\labelsof{\alpha}{\T}}$. By definition, $\phi$ cannot be an equation.
If $\phi$ contains a $\depth$-atom then by Lemma \ref{lem:depthatom} it must be of the form
$\depth(\paramA) \ordrel \n$, where $\ordrel \in \{ \less, \iseq \}$.
Otherwise, $\phi$ must be a subformula introduced either by \startrule or by \inductionrule (up to a renaming of parameters).
Hence $\phi$ must be in $\classform$.
If $\phi$ is not a \bform then \veerule or \wedgerule applies.

Finally, if $\paramA$ and $\paramB$ are two parameters occurring in $\labelsof{\alpha}{\T}$ then by irreducibility w.r.t. \seprule, either $\paramA \iseq \paramB$ (or $\paramB \iseq\paramA$) occurs in $\labelsof{\alpha}{\T}$ (in which case $\paramA$ or $\paramB$ is \solved) or $\paramA \not \iseq \paramB \in \labelsof{\alpha}{\T}$.
\end{proof}
}

\addtheowithlemmas{
\begin{theorem}
\label{lem:comp}
Let $\T$ be a \prooftree. If $\alpha$ is a node in $\T$
that is irreducible by all the expansion rules then
$\labelsof{\alpha}{\T}$ is $\classint$-satisfiable iff  $\noneq{\labelsof{\alpha}{\T}}$ is.
Furthermore, $\noneq{\labelsof{\alpha}{\T}}$ is a set of \bforms.
\end{theorem}
}
{Theorem \ref{lem:comp}}
{
\protect\complem
}
{\protect\prcomp}

\proofstate{\subsection{Termination}}{}

\newcommand{\mes}{\mathrm{\it mes}}
\newcommand{\sizeof}{\mathrm{\it weight}}
\newcommand{\nbparam}{\mathrm{\it separable}}
\newcommand{\nbdec}{\mathrm{\it diseq}}

\newcommand{\termaux}{
We define the following measures on formul{\ae}:

\begin{definition}
\label{def:sizeof}
Let $a$ be the maximal arity of the symbols in $\Sigma$.
We denote by $\sizeof$ a function mapping every term, atom or literal to a natural number, defined as follows:
\begin{enumerate}
\item{$\sizeof(\paramA) = 1$ if $\paramA$ is a parameter.}
\item{
$\sizeof(f(t_1,\dots,t_n)) \isdef 1+\Sigma_{i=1}^n \sizeof(t_i)$ if $f \not = \depth,\su$.}
\item{$\sizeof(\depth(t)) \isdef \sizeof(t)$.}
\item{$\sizeof(\su(t)) \isdef 3+a+\sizeof(t)$.}
\item{$\sizeof(t \iseq s) = \sizeof(t \less s) = \sizeof(t)+\sizeof(s)+1$.}
\item{$\sizeof(\neg \phi) \isdef 1+\sizeof(\phi)$.}
\item{
$\sizeof(t \lesseq s) = \sizeof(t)+\sizeof(s)+2$}
\item{
$\sizeof(\appred{\dpred}{\paramA}) = 1+  \max_{f \in \Sigma_{\sort}} \sizeof(\psi_f)$,
where $\psi_f$ is obtained from $\predeft{\dpred}{f(\vec{\paramB})}$ by replacing every occurrence of $f(\vec{\paramB})$ by $\paramA$.
$\vec{\paramB}$ denotes a vector of parameters of the same sort as the domain of $f$ (the value of $\sizeof$ does not depend on the names of the parameter, thus they can be chosen arbitrarily).}
\end{enumerate}
\end{definition}

\newcommand{\nbsolved}{\mathrm{\it unsolved}}

\begin{definition}
\label{def:mes}
$$\mes(S) \isdef (\{ \sizeof(\phi) \mid \phi \in S' \},\nbparam(S), \nbdec(S), \nbsolved(S))$$
where:
\bitem
\item{$S'$ denotes the set of formul{\ae} in $S$ that are not of the form $\paramA \iseq \paramB$ or $\paramA \not \iseq \paramB$, with $\paramA,\paramB \in \param$.}
\item{$\nbparam(S)$ denotes the number of pairs of parameters $(\paramA,\paramB)$ occurring in $S$ such that neither $\paramA\iseq \paramB$ nor $\paramA \not \iseq \paramB$ is contained in $S$.}
\item{$\nbdec(S)$ denotes the number of formul{\ae} in $S$ on which \diseqrule applies.}
\item{$\nbsolved(S)$ is the number of unsolved parameters in $S$.}
\eitem
The measure $\mes$ is ordered by the lexicographic and multiset extensions of the usual ordering on natural numbers.
\end{definition}

The next lemma shows that all the expansion rules, except \depthexplrule, strictly decrease  $\mes$ (possibly after some applications of the decomposition rules):

\begin{lemma}
\label{lem:dec}
Let $\T$ be a \prooftree. If $\alpha$ is a node obtained from a node $\beta$ by applying an expansion rule distinct from
\depthexplrule, then there exists a node $\alpha'$ such that $\alpha \parentd{\T}^* \alpha'$ and $\mes(\labelsof{\alpha'}{\T}) < \mes(\labelsof{\beta}{\T})$.
\end{lemma}

\begin{proof}
We distinguish several cases.
\bitem
\item{\decomprules. The rules \veerule, \wedgerule, \clashrule, \nclashrule and \eclashrule remove (at least) one logical symbol from $S$, thus $\sizeof$ decreases strictly. }
\item{\inductionrules.
The rule replaces a formula $\appred{\dpred}{\paramA}$ by
a formula $\psi$ obtained from $\predeft{\dpred}{f(\paramB_1,\dots,\paramB_n)}$ by replacing $f(\paramB_1,\dots,\paramB_n)$ by $\paramA$.
By definition of $\sizeof$, we have
$\sizeof(\appred{\dpred}{\paramA}) > \sizeof(\psi)$, thus $\mes$ decreases strictly.}
\item{\equalityrules.
\bitem
\item{\Eqrule.
The rule temporality increases $\sizeof$, since a (complex) formula $\psi = \decompe{f(\paramA_1,\dots,\paramA_n)}{g(\paramB_1,\dots,\paramB_m)}$ is added in $S$. However, due to the control, the decomposition rules must be immediately applied on this formula. By definition, $\decompe{f(\paramA_1,\dots,\paramA_n)}{g(\paramB_1,\dots,\paramB_m)}$ only contains the symbols $\vee$, $\wedge$, $\iseq$ and parameters in $\paramA_1,\dots,\paramA_n,\paramB_1,\dots,\paramB_m$, thus it must be reduced by decomposition into equations between parameters. Thus $\sizeof$ cannot increase. Furthermore, since an equation $A \iseq g(\paramB_1,\dots,\paramB_m)$ is deleted, $\sizeof$ must decrease.}
\item{\Diseqrule. The rule temporality increases $\sizeof$, since a (complex) formula $\psi$ is added in $S$. However, due to the control, the decomposition rules must be immediately applied on this formula. By definition, $\decompe{f(\paramA_1,\dots,\paramA_n)}{g(\paramB_1,\dots,\paramB_m)}$ only contains the symbols $\vee$, $\wedge$, $\iseq$ and parameters in $\paramA_1,\dots,\paramA_n,\paramB_1,\dots,\paramB_m$. Thus $\NNF(\neg \psi)$, being the nnf of $\neg \decompe{f(\paramA_1,\dots,\paramA_n)}{g(\paramB_1,\dots,\paramB_m)}$, must be reduced by decomposition into disequations between parameters. Thus $\sizeof$ cannot increase. Obviously, $\nbparam$ does not increase either and $\nbdec$ decreases, by definition.
    }
\item{\Seprule. It is clear that the rule does not increase $\sizeof$ (since only equations or disequations  between parameters are added) and decreases $\nbparam$.}
\item{\Reprule. The rule does not increase $\sizeof$, $\nbparam$ and $\nbdec$ and decreases $\nbsolved$.}
\eitem}
\item{\depthrules.
\bitem
\item{\Strictrule:
A formula $\depth(\paramA) \lesseq \n$ is replaced by
$\depth(\paramA)\iseq \n \vee \depth(\paramA) \less \n$. After decomposition, this last formula is reduced to either
$\depth(\paramA)\iseq \n$ or $\depth(\paramA) \less \n$.
We have $\sizeof(\depth(\paramA)\iseq \n) = \sizeof(\depth(\paramA)\less \n) = 3$ and
$\sizeof(\depth(\paramA)\lesseq \n) = 4$. Thus $\sizeof$ decreases.
}
\item{\Lessclashrule: Since the rule only adds a disequation between parameters, $\sizeof$ does not increase. Moreover, $\nbparam$ decreases, due to the control.}
\item{\Lessrule: A formula $\depth(\paramA) \less \su(\n)$ is replaced by
$\depth(\paramA) \lesseq \n$. We have $\sizeof(\depth(\paramA) \less \su(\n)) = 6+a$ and $\sizeof(\depth(\paramA) \lesseq \n) = 5$. Thus $\sizeof$ decreases.}
\eitem
}
\item{\Explrule. After decomposition, a formula of the form $\depth(\paramB) \iseq \su(t)$ is replaced by formul{\ae} of the form $\paramB \iseq t_i$ or $\depth(\paramA) \iseq t$ or $\depth(\paramA) \less t$.
    We have $\sizeof(\depth(\paramB) \iseq \su(t)) = 4+a+\sizeof(t)$ and
    $\sizeof(\paramB \iseq t_i) = 2+\sizeof(t_i) \leq 3+a$, $\sizeof(\depth(\paramA) \iseq t) = 2+\sizeof(t)$,
      $\sizeof(\depth(\paramA) \lesseq t) = 3+\sizeof(t)$.
      Thus $\sizeof$ decreases.
     }
\eitem

\end{proof}

A {\em \layerform} is a set of formul{\ae} that is irreducible w.r.t.
all expansion rules, except \depthexplrule.

We now prove that $\moregen$ is a well quasi-order for \layerforms.
We need to introduce some additional definitions.
A sequence of sets of formul{\ae} $(\Phi_i)_{i \in [0,n|}$ (with $n \in \Ninf$) is {\em \bad} iff there are no indices $i,j \in [0,n[$ such that $i < j$ and $\Phi_j \moregen \Phi_i$.
A \layerform $\Phi$ is \emph{built on} a set of \bforms $\Gamma$ iff all the non equational formul{\ae} in $\Phi$ are of the form $\replaceby{\phi}{x}{\paramA}$, where $\phi \in \Gamma \cup \{ \depth(x) \less \n, \depth(x) \iseq \n \}$ ($x$ is a variable and $\paramA$ a parameter).

\begin{proposition}
Let $\T$ be a \prooftree for a formula $\phi$.
There exists a finite set of \bforms $\Gamma$ such that for every layer $\alpha$, $\noneq{\labelsof{\alpha}{\T}}$ is built on $\Gamma$.
\end{proposition}

\begin{proof}
Let $\alpha$ be a \layer in $\T$.
By Lemma \ref{lem:depthatom}, all the formul{\ae} containing $\depth$ must be of the form $\depth(\paramA) \ordrel \n$ where $\ordrel \in \{ \iseq, \less \}$.
By Lemma \ref{lem:layerform}, the remaining formul{\ae} must be \bforms.
The only rules that can add new \bforms into the \prooftree (up to a renaming of parameters)
are \startrule and \inductionrule.
The former only adds \bforms occurring in $\phi$.
The formul{\ae} introduced by the latter rule are of obtained from formul{\ae} occurring in $\rules$ by
instantiating variables by constant symbols and
replacing a term $f(\vec{\paramA})$ by a parameter.
By definition there are only finitely many such formul{\ae} (up to a renaming of parameters).
\end{proof}

\begin{lemma}
\label{lem:well}
Let $(\Phi_i)_{i\in \N}$ be an infinite sequence of \layerforms built on a given finite set of \bforms $\Gamma$.
 The sequence $(\noneq{\Phi_i})_{i \in \N}$ is \good.
\end{lemma}

\newcommand{\subsetgamma}{\Lambda}

\begin{proof}
Let $\Psi_i = \noneq{\Phi_i}$.
By Lemma \ref{lem:layerform}, $\Psi_i$
contains only \bforms, formul{\ae} of the form $\appred{\dpred}{\paramA}$ or $\depth(\paramA) \ordrel \n$ and disequations between parameters.

For every parameter $\paramA$, we denote by $\propof{\paramA}{\Psi_i}$ the set of \bforms $\psi \in \Gamma \cup \{ \depth(x)\less \n, \depth(x)\iseq \n \}$ containing a variable $x$ and such that $\replaceby{\psi}{x}{\paramA} \in \Psi_i$. We write $\paramA \similar{\Psi_i} \paramB$ iff
$\propof{\paramA}{\Psi_i} = \propof{\paramB}{\Psi_i}$. The relation $\similar{\Psi_i}$ is obviously an equivalence relation.
For every $\subsetgamma \subseteq \Gamma \cup \{ \depth(x)\less \n, \depth(x)\iseq \n \}$, we denote by $\paramof{\subsetgamma}{\Psi_i}$ the set of parameters $\paramA$ such that $\propof{\paramA}{\Psi_i} = \subsetgamma$.

For any sequence $\Psi = (\Psi_i)_{i \in I}$ of sets of formul{\ae},
we denote by $\propin{\Psi}$ the set $\{ \propof{\paramA}{\Psi_i} \mid \paramA \in \param, i \in I \}$.
Note that $\propin{\Psi} \subseteq 2^{\Gamma \cup \{\depth(x) \less \n, \depth(x) \iseq \n\}}$.

Assume that $\Psi = (\Psi_i)_{i \in \N}$ is \bad.
Without loss of generality, we assume that $\propin{\Psi}$ is minimal,
i.e. if $\Psi'$ is a sequence of
\bforms built on $\Gamma$ such that
$\propin{\Psi'} \subset
\propin{\Psi}$, then $\Psi'$ is
\good.

If $\propin{\Psi} = \emptyset$ then necessarily, the $\Psi_i$'s ($i \in I$)
contain only sets of formul{\ae} in $\Gamma$ and disequation between parameters.
Since $\Gamma$ is finite, the number of sets of formul{\ae} in $\Gamma$ is also finite. Since $\Psi_i$ is infinite, there exists some subsequence $\Psi' = \Psi'_{i \in \N}$ of $\Psi$ such that for every $i,j \in \N$, $\Psi'_i$ and $\Psi'_j$ only differ by disequations.
Let $i \in I$ be the index in $\N$ such that the number of parameters in $\Psi'_i$ is minimal. We show that $\Psi'_{i+1} \moregen \Psi'_{i}$.
By definition the number of parameters occurring in $\Psi'_{i+1}$ is greater or equal to that of $\Psi'_i$. Thus there exists an injective function $\ren$ from the set of parameters in $\Psi'_{i}$ onto the set of parameters in $\Psi'_{i+1}$. Then, if $\paramA \not \iseq \paramB$ is a formula in $\Psi'_i$, there must exist two distinct parameters, $\paramA',\paramB'$ such that
$\ren(\paramA) = \paramA'$ and $\ren(\paramB) = \paramB'$, and $\paramA',\paramB'$ occurs in $\Psi'_{i+1}$. Furthermore, by Lemma \ref{lem:layerform}, $\paramA' \not \iseq \paramB'$ occurs in $\Psi'_{i+1}$. Thus $\ren(\Psi'_{i}) \subseteq \Psi'_{i+1}$, whence $\Psi'_{i+1} \moregen \Psi'_i$. This means that $\Psi'$ (hence also $\Psi$) is \good, which contradicts our hypothesis.

Thus we must have $\propin{\Psi} \not = \emptyset$.
Let $\paramA$ be a parameter occurring in $\Psi_0$ and let
$\subsetgamma \iseq \propof{\paramA}{\Psi_0}$.
Let $k \in [0,|\paramof{\subsetgamma}{\Psi_0}|]$.
Consider the set of indices $\{i_j \mid j \in \N \}$ such that
$|\paramof{\subsetgamma}{\Psi_{i_j}}| = k$.
Let $\Psi' = \Psi'_{j \in I'}$ be the sequence such that
$\Psi'_j$ is obtained from $\Psi_{i_j}$ by removing each
formul{\ae} $\phi$ containing a parameter $\paramA \in \paramof{\subsetgamma}{\Psi_{i_j}}$.
By definition of $\Psi'$, we have
$\propin{\Psi'} \subset \propin{\Psi}$ (since $\propin{\Psi'}$  cannot contain $\subsetgamma$).

Assume that $\Phi'$ is infinite.
By minimality of $\Psi$, $\Psi'$ must be \good.
Consequently, there exist two indices $j < j'$ such that
$\Psi'_{j'} \moregen \Psi'_{j}$, i.e.
there exists a renaming $\ren$ such that
$\ren(\Psi'_{j}) \subseteq \Psi'_{j'}$.
By definition of $\Psi_{i_j}$, we have
$|\paramof{\subsetgamma}{\Psi_{i_j}}| = |\paramof{\subsetgamma}{\Psi_{i_{j'}}}| = k$.
Let $\ren'$ be any bijective renaming from
$\paramof{\subsetgamma}{\Psi_{i_j}}$ to $\paramof{\subsetgamma}{\Psi_{i_{j'}}}$.
By definition of $\Psi'$, $\ren'$ and $\ren$ must have disjoint domains. Let $\ren'' = \ren \cup \ren'$.
It is easy to check that we have $\ren''(\Psi_{i_{j}}) \subseteq \Psi_{i_{j'}}$,
hence $\Psi_{i_{j'}} \moregen \Psi_{i_j}$, which is impossible.
Thus $\Psi'$ is finite.
Since this is true for every $k \leq |\paramof{\subsetgamma}{\Psi_0}|$, this implies that there exists some index $j$ such that for every $i \geq j$, we have $|\paramof{\subsetgamma}{\Psi_i}| \geq |\paramof{\subsetgamma}{\Psi_0}|$.
But then, since the same reasoning holds for every $\subsetgamma$,
there must exist some $j \in I$ such that for every $i \geq j$ and for every $\subsetgamma \subseteq \Gamma \cup \{ \depth(x) \iseq \n, \depth(x) \lesseq \n \}$:
$|\paramof{\subsetgamma}{\Psi_i}| \geq |\paramof{\subsetgamma}{\Psi_0}|$ (it suffice to take the maximal value of all the $j$'s corresponding to each $\subsetgamma$, which is possible since the number of distinct set $\subsetgamma$ is finite).

We have in particular:
$\forall \subsetgamma \subseteq \Gamma \cup \{ \depth(x) \iseq \n, \depth(x) \lesseq \n \},\, |\paramof{\subsetgamma}{\Psi_i}| \geq |\paramof{\subsetgamma}{\Psi_0}|$.

Thus there exists an injective renaming
$\ren_\subsetgamma$ from $\paramof{\subsetgamma}{\Psi_0}$ to $\paramof{\subsetgamma}{\Psi_i}$.
By definition, if $\subsetgamma \not = \subsetgamma'$, then $\ren_\subsetgamma$ and $\ren_{\subsetgamma'}$ have disjoint domains. Let $\ren = \bigcup_{\subsetgamma \subseteq \Gamma \cup \{ \depth(x) \iseq \n, \depth(x) \lesseq \n \}} \ren_\subsetgamma$.
It is clear that $\ren(\Psi_0) \subseteq \Psi_i$. Thus $\Psi$ is \good.

\end{proof}

}

\newcommand{\termmain}{
Assume that there exists an infinite  \prooftree $\T$.
$\T$ must have at least one infinite branch $(\alpha_i)_{i \in \N}$.
If there exist $i,j \in \N$ such that
$i < j$ and $\labelsof{\alpha_j}{\T} \moregen \labelsof{\alpha_i}{\T}$ then \looprule applies on $j$, which is impossible.
Thus, by Lemma \ref{lem:well}, the subsequence of \layerforms in $\labelsof{\alpha_i}{\T}_{i \in \N}$ is finite, and
there exists $i \in \N$ such that for every $j \geq i$, $\alpha_j$ is not a \layerform. In this case, \depthexplrule
cannot be applied on $\alpha_j$, hence by Lemma \ref{lem:dec}, we have
$\mes(\labelsof{\alpha_{j+l}}{\T}) < \mes(\labelsof{\alpha_j}{\T})$, for some $l > 0$.
Since $\mes$ is well-founded, we get a contradiction.
}

We finally state
that the procedure is terminating.
\addtheowithlemmas{
\begin{theorem}
\label{theo:term}
The expansion rules terminate on every formula in $\classform$.
\end{theorem}
}
{Theorem \ref{theo:term}}
{\protect\termaux}
{\termmain}

\begin{corollary}
If the satisfiability problem is decidable (resp. semi-decidable) for \bforms in $\classform$ then it is so for all formul{\ae} in $\classform$.
\end{corollary}

\section{Conclusion}

\label{sect:conc}

We have proposed a proof procedure for reasoning on schemata of formul{\ae} (defined by induction on an arbitrary structure, such as natural numbers, lists, trees etc.) by
relating the satisfiability problem for such schemata to
that of a \emph{finite} disjunction of formul{\ae} in the base language.
Our approach applies to a wide range of formul{\ae}, which may be interpreted in some specific class of structures (e.g. arithmetics).
It may be seen as a generic way to add inductive capabilities
into logical languages, in such a way that the main computational properties of the initial language (namely decidability or semi-decidability) are preserved.
To the best of our knowledge, no published procedure offers similar features.
There are very few decidability or even completeness results in inductive theorem proving and we hope that the present work will help to promote new progress in this direction.
Future work includes the implementation of the proof procedure and its extension to non-monadic defined symbols.


\docend

\end{document}